\documentclass{article}
\usepackage{graphicx} 
\usepackage[]{amsmath,amssymb,amsfonts,latexsym,amsthm,enumerate,fullpage,xcolor,bbm}
\usepackage{algorithm}
\usepackage{algorithmic}
\usepackage{comment}
\usepackage{nicefrac}
\usepackage{pdflscape}
\usepackage{tikz}
\usepackage{tikz-cd}
\usepackage{stmaryrd}
\usepackage{comment}
\usepackage{xcolor}

\usepackage[
backend=biber,
style=alphabetic,
sorting=nyt,
maxnames=7,
maxalphanames=7,
backref=true,
url=false
]{biblatex}
\usepackage[colorlinks=true,citecolor=blue,linkcolor=blue]{hyperref}
\hypersetup{
    colorlinks,
    linkcolor={red!50!black},
    citecolor={blue!50!black},
    urlcolor={blue!80!black}
}
\usepackage[nameinlink, noabbrev, capitalize]{cleveref}

\newtheorem{counter}{Counter}[section]
\newtheorem{theorem}[counter]{Theorem}
\newtheorem{thm}{Theorem}
\newtheorem{lemma}[counter]{Lemma}

\newtheorem{claim}[counter]{Claim}
\newtheorem{fact}[counter]{Fact}

\newtheorem{corollary}[counter]{Corollary}

\newtheorem{definition}[counter]{Definition}

\newtheorem{observation}[counter]{Observation}
\newtheorem{remark}[counter]{Remark}

\newcommand{\seed}{\psi}
\newcommand{\N}{\mathbb{N}}

\newcommand{\F}{\mathbb{F}}

\newcommand{\R}{\mathbb{R}}
\newcommand{\poly}{\operatorname{poly}}

\newcommand{\eps}{\varepsilon}
\newcommand{\abs}[1]{\left\vert#1\right\vert}

\usepackage{array}
\newcommand{\PreserveBackslash}[1]{\let\temp=\\#1\let\\=\temp}
\newcolumntype{C}[1]{>{\PreserveBackslash\centering}p{#1}}
\newcolumntype{R}[1]{>{\PreserveBackslash\raggedleft}p{#1}}
\newcolumntype{L}[1]{>{\PreserveBackslash\raggedright}p{#1}}

\usepackage{pifont}

\newcommand{\GammaL}[1][\relax]{\Gamma_{\shortrightarrow}{#1}}
\newcommand{\GammaR}[1][\relax]{\Gamma_{\shortleftarrow}{#1}}

\title{Two-Sided Lossless Expanders in the Unbalanced Setting}
\author{}
\author{  Eshan Chattopadhyay\thanks{Supported by a Sloan Research Fellowship and NSF CAREER Award 2045576.}\\ Cornell University\\ \texttt{eshan@cs.cornell.edu}  \and Mohit Gurumukhani\footnotemark[1] \\ Cornell University\\ \texttt{mgurumuk@cs.cornell.edu} \and  Noam Ringach \thanks{Supported by NSF GRFP grant DGE – 2139899,  NSF CAREER Award 2045576 and a Sloan Research Fellowship.}   \\ Cornell University\\ \texttt{nomir@cs.cornell.edu} \and Yunya Zhao\footnotemark[1] \\ Cornell University\\ \texttt{yunya@cs.cornell.edu}  }
\date{}

\begin{document}
\maketitle

\begin{abstract}

    We present the first explicit construction of \emph{two-sided} lossless expanders in the unbalanced setting  (bipartite graphs that have polynomially many more nodes on the left than on the right).  
    Prior to our work, all known explicit constructions in the unbalanced setting achieved only one-sided lossless expansion.

     Specifically, we show that the one-sided lossless expanders constructed by Kalev and Ta-Shma (RANDOM'22)---that are based on multiplicity codes introduced by Kopparty, Saraf, and Yekhanin (STOC'11)---are, in fact,  two-sided lossless expanders. Moreover, we show that our result is tight, thus completely characterizing the graph of Kalev and Ta-Shma.

     Using our unbalanced bipartite expander, we easily obtain lossless (non-bipartite) expander graphs on $N$ vertices with polynomial degree $\ll N$ and expanding sets of size $N^{0.49}$.
\end{abstract}

\section{Introduction}

Lossless expanders are graphs in which small sets of vertices have almost as many neighbors as possible. Formally, we say that a $d$-regular graph $G=(V,E)$ is a \emph{$(K,A)$-expander} if for all sets $S\subseteq V$ of size at most $K$ we have that $\abs{\Gamma(S)}\geq A\abs{S}$ where $\Gamma(S)$ is the neighborhood of $S$. Generally, we desire that $K$ is as large as possible with $K=\Omega(\abs{V}/d)$. When $A=(1-\varepsilon)d$ for some small $\varepsilon$, we say that $G$ is a \emph{$(K,\varepsilon)$-lossless expander} since only a small fraction of the total number of possible neighbors is lost. It is well-known that a random $d$-regular graph is a $(K=\gamma n,\varepsilon=0.01)$-lossless expander with high probability, for some constant~$\gamma > 0$.

A reasonable question after seeing this definition is whether other notions of expansion, such as spectral or edge expansion, can be used to derive such graphs. Unfortunately, while Ramanujan graphs (optimal spectral expanders) do have expansion factor arbitrarily close to $A=d/2$, there also exist examples of Ramanujan graphs with expansion factor exactly  $A=d/2$, showing that spectral expansion does not necessarily give rise to lossless expansion \cite{kahale_eigenvalues_1995}.

The study of lossless expanders has paid special attention to bipartite graphs due to their connection with randomness condensers. A \textit{one-sided lossless} expander is a bipartite graph $G= (L \sqcup R, E)$ where every ``small enough'' set on the left expands losslessly to the right. It is standard to view lossless condensers and one-sided bipartite lossless expanders as related objects. For this purpose, it is natural to talk about highly unbalanced bipartite graphs in which $\abs{L} \gg \abs{R}$, that is, the neighbor function that takes in a left vertex and the index of a neighbor and outputs the right vertex has a much shorter output length than input length. Current explicit constructions for unbalanced one-sided lossless expanders \cite{ta-shma_better_2006, ta-shma_lossless_2007, guruswami_unbalanced_2009, kalev_unbalanced_2022}—with the best parameters achieved by \cite{guruswami_unbalanced_2009, kalev_unbalanced_2022}
\footnote{The lossless expanders of \cite{kalev_unbalanced_2022} have slightly better dependence on constants compared to \cite{guruswami_unbalanced_2009}.}
—have found a wide array of applications in coding theory \cite{sipser_expander_1996}, extractor constructions \cite{ta-shma_better_2006, ta-shma_lossless_2007, guruswami_unbalanced_2009, dvir_extensions_2013}, derandomization \cite{doron_derandomization_2023},\footnote{\cite{doron_derandomization_2023} instantiated Goldreich's PRG \cite{goldreich_world_2011} with the lossless expander of \cite{kalev_unbalanced_2022}.} and probabilistic data structures \cite{upfal_how_1987, buhrman_are_2002}, with the unbalanced nature of the graph being essential. 

We say a bipartite graph is a \textit{two-sided lossless} expander if lossless expansion happens in both directions. Intuitively, in the unbalanced case where $\abs{L} \gg \abs{R}$, lossless expansion is hard to achieve from left to right because there is little room on the smaller side to allow for expansion from the much larger side; on the other hand, the right-to-left expansion seems to be much easier at first sight for the same reason. However, despite aforementioned constructions for one-sided lossless expanders as well as their broad applications, there has been no known explicit construction of unbalanced two-sided lossless expanders before this work. Therefore, it is an extremely intriguing direction to further the theory of lossless expanders. We are not aware of any existing applications of two-sided lossless expanders with a polynomial imbalance between the left and right—we leave this as an interesting open problem—but we believe that a deeper understanding of the structure of these objects is the first step towards making connections with other topics and finding new applications.

In this work, we try to fill the gap, namely, the lack of explicit unbalanced two-sided lossless expanders by closely studying the bipartite graph of \cite{kalev_unbalanced_2022} based on the multiplicity codes of \cite{kopparty_high-rate_2014}.
\footnote{It is natural to consider whether the unbalanced expander from \cite{guruswami_unbalanced_2009} is a two-sided lossless expander. It will be interesting to determine this since the GUV graph is not even right-regular---see \cref{sec:GUV not right regular} for details.} 
For simplicity, we refer to this graph as the ``KT graph'' (see \cref{def: KTGraph full}). 
As our first main result, we show that the KT graph is a two-sided lossless expander. Moreover, we prove that the expansion is tight up to a constant multiplicative factor. 
As our second main result, we obtain non-bipartite lossless expanders with high degree by taking the bipartite half of the KT graph.

We note that  lossless expanders have been extensively studied in the \textit{balanced} setting as well. We discuss this in \cref{sec: balanced lossless}.

\subsection{Our results}
We first define a two-sided lossless expander formally:
\begin{definition}[Two-sided lossless expander]\label{def:intro-two-sided lossless expander}
    We say that a $(D_L,D_R)$-regular bipartite graph $G=(L\sqcup R,E)$ is a \emph{two-sided $(K_L,A_L,K_R,A_R)$-expander} if for any subset $S_L\subseteq L$ such that $\abs{S_L}\leq K_L$ we have that $\abs{\GammaL(S_L)}\geq A_L\abs{S_L}$ and similarly that for any subset $S_R\subseteq R$ such that $\abs{S_R}\leq K_R$ we have that $\abs{\GammaR(S_R)}\geq A_R\abs{S_R}$. When $A_L=(1-\varepsilon_L)D_L$ and $A_R=(1-\varepsilon_R)D_R$ for small $\eps_L, \eps_R > 0$, we say that $G$ is a \emph{two-sided $(K_L,\varepsilon_L,K_R,\varepsilon_R)$-lossless expander}. 
\end{definition}

With the above definition, we are ready to state the main theorems:

\begin{thm}[Informal version of \cref{thm: main simplified theorem}, bipartite two-sided lossless expander]\label{thm: informal bipartite pos and impos}  
    For infinitely many $N$ and all constant $0 < \delta \le 0.99$, there exists an explicit, biregular, two-sided $(K_{L}, \eps_L = 0.01, K_{R}, \eps_{R} = 0.01)$  lossless expander $\GammaL: [N]\times [D_L]\to [M]$ where $D_L = \poly(\log N)$, $N^{1.01\delta - o(1)}\le M \le D_L\cdot N^{1.01 \delta}$, $K_{L} = N^{\delta}$, and $K_{R} = \min\left(O(M / D_L), O(N / (MD_L))\right)$. 
    
    On the other hand, for every such graph, there exists a set $S \subseteq R$, $\abs{S} = O(K_R)$, on the right side which has less than $1/2 \cdot D_R\cdot \abs{S}$ neighbors on the left.
\end{thm}

\begin{remark}\label{rmk:tightness when n>2s+2}
    Because \cite{kalev_unbalanced_2022} has optimal left degree  of their bipartite graph (up to polynomial factors), we achieve optimal left-degree as well and, with respect to this, achieve optimal right degree, optimal expansion constant, optimal size of sets of vertices on left side that losslessly expand, and optimal size of sets of vertices on right side that losslessly expand when $M\le \sqrt{N}$. \footnote{To see that this setting of $K_R$ is indeed optimal, note that in a $(D_L,D_R)$-biregular graph it must be that $N\cdot D_L=M\cdot D_R$ and so $\frac{M}{D_L} = \frac{N}{D_R}$. Hence, $K_R = O(M/D_L) = O(N / D_R)$, the largest possible size.}
\end{remark}

We obtain our second main result by taking the bipartite half (see \cref{subsec: proof overview non-bipartite} for more details) of the \cite{kalev_unbalanced_2022} graph, using the fact that it is a two-sided lossless expander:

\begin{thm}[Informal version of \cref{thm:normal lossless expander}, non-bipartite lossless expander] \label{thm: informal non-bipartite pos and impos}
    For infinitely many $N$ and all constant $0 < \delta < 0.99$, there exists an explicit regular $(K, \eps = 0.01)$ lossless expander $G = (V, E)$ where $|V| = N$, the degree is $D$ where $N^{1 - 1.01\delta}\le D \le N^{1 - 1.01\delta + o(1)}$ and $K = \min\left(N^{\delta}, N^{1 - 1.01\delta - o(1)}\right)$. 
    Furthermore, $G$ with  one vertex is removed, is endowed with a free group action from the multiplicative group $\F_q$, where $q=\poly(\log N)$. 
\end{thm}

We realize that though the free group action adds to the structure of this graph, the group action is too small for applications. 
\begin{remark}
When $\delta \le 0.49$, the value of $K$ is almost optimal (a trivial upper bound is $K\le N / D$) since in that regime, $K = N^{\delta} \ge \left(N / D\right)^{0.99}$.
\end{remark}

One can show that there exist non-bipartite lossless expanders with even constant degree. So, the degree of our lossless graph obtained is far from optimal. Nevertheless, as far as we know, this is the first explicit construction of a regular lossless (non-bipartite) expander with expanding sets of size $N^{0.49}$. 

\subsection{Lossless expanders in the balanced setting} \label{sec: balanced lossless}
A parallel line of work considers \textit{balanced} bipartite lossless expanders, with motivations from coding theory. There have been explicit constructions of balanced one-sided lossless expanders \cite{capalbo_randomness_2002, cohen_hdx_2023, golowich_new_2024}, with which one can construct good error correcting codes \cite{sipser_expander_1996}. There has also been progress in understanding two-sided lossless expanders in the balanced setting. Lossless expansion was shown to be feasible in high-girth regular graphs \cite{MM21girthramanujan, hsieh_explicit_2024-1}, taking the bipartite double cover of which implies a balanced two-sided lossless expander. It was shown in \cite{lin_good_2022} that balanced two-sided lossless expanders with constant degree, constant imbalance, and certain algebraic properties have applications to good quantum low-density parity check (qLDPC) codes. Towards this direction, \cite{hsieh_explicit_2024-1} constructed explicit two-sided lossless expansion for extremely small sets of size $K=\Omega(\exp(\sqrt{\log\abs{V}}))$, which is still not sufficient for the application in \cite{lin_good_2022}. Our work in the polynomially unbalanced setting is incomparable to the balanced setting, as it is possible to achieve constant degree in the balanced setting, while having a polynomial imbalance in left and right nodes forces the graph to have non-constant degrees.

\paragraph{Concurrent works} Since our paper was made public online, there have been two new papers on constructions of balanced lossless expanders. Chen \cite{chen25unique} achieved balanced two-sided lossless expansion for polynomial-sized sets as one of their several results; however, their size of expanding sets is not optimal in the balanced setting. Our analysis proves optimal expanding set size in the unbalanced setting (see \cref{rmk:tightness when n>2s+2}). 
Hsieh, Lin, Mohanty, O'Donnell and Zhang  \cite{hsieh_explicit_2024} achieved $\frac{3}{5}$-two-sided unique-neighbor expansion in the balanced setting, which is the first explicit construction of balanced two-sided vertex expanders beyond the spectral barrier.

\section{Proof Overview}

In this section, we first outline the proof of \cref{thm: informal bipartite pos and impos}—our two-sided lossless expander. Using it, we construct high degree non-bipartite lossless expanders, proving \cref{thm: informal non-bipartite pos and impos}.

\subsection{Two-sided lossless expander}

We show that the bipartite graph defined in \cite{kalev_unbalanced_2022} based on multiplicity codes is a two-sided lossless expander. The left-to-right lossless expansion was shown in \cite{kalev_unbalanced_2022}. Our main contribution  is showing that the KT graph also expands losslessly from right to left. To do this, we first show that the KT graph is right-regular. Second, for any pair of right vertices, we compute the exact number of common left neighbors they have. Finally, for any not-too-large subset on the right, we lower bound the number of its left neighbors by using the inclusion-exclusion principle  to subtract all possible double counted common left neighbors from the total number of outgoing edges. 
\par We state an informal version of our result and present details on the strategy sketched above.
\begin{theorem}[Informal version of \cref{MainThm}]\label{MainThm-Informal}
            For every field $\F_q$ and $n, s\in \N$ with $15\le (s+1) < n < char(\F_q)$, and any $\delta>0$, there exists an explicit bipartite graph $G = (L\sqcup R, E)$ with $L = \F_q^n, R = \F_q^{s+2}$ with left degree $d_L=q$ and right degree $d_R= q^{n - (s+1)}$ such that $G$ is a
            two-sided $(K_L, A_L, K_R, A_R)$ expander
            where $K_L=\Omega(q^{s+1})$, $A_L = q - n(s+2)$, $K_R=\delta q^{\min(s+1, n - (s+1))}$, $A_R = \left(1 - O\left(\delta\cdot\frac{q-1}{q}\right)\right)q^{n - (s+1)}$.
\end{theorem}
 \cref{thm: informal bipartite pos and impos} is obtained from \cref{MainThm-Informal} by instantiating the parameters appropriately (see \cref{sec: plug in param} for more details). We now define the KT graph and then claim that it is a lossless right expander.
\begin{definition}[The KT graph \cite{kalev_unbalanced_2022}]\label{def: KTGraph}
    Let $q, n, s\in \N$ be such that $q$ is a prime power, characteristic of the finite field $\F_q \ge n$ and $s \le n / 2$. We define $G = (L \sqcup R, E)$ where $L = \F_q^n, R = \F_q^{s+2}$. The left degree is $q$ and for any $f\in \F_q^n$ and $y\in \F_q$, the $y$'th neighbor of $f$ is defined as follows: Identify $f$ as member of $\F_q[X]$ with degree of $f$ at most $n-1$ ; then, the neighbor $\GammaL{(f,y)}$ will be $(y, f^{(0)}(y), \dots, f^{(s)}(y))$ where $f^{(i)}$ is the $i$'th iterative derivative of $f$.
\end{definition}

\begin{theorem}[The KT graph losslessly expands from the right]\label{thm: right sided lossless expansion- Informal}
        The KT graph $G$ is a right $(K_R,A_R)$-lossless expander where $K_R=\delta \min(\abs{R}, \abs{L} / \abs{R})$, $\eps_R=O(\delta\cdot\frac{q-1}{q})$ for arbitrary $0 < \delta < 1$. In other words, for any subset $T\subseteq R$, $\abs{T} \leq K_R$, $T$ has at least $(1-\eps_R)d_R|T|$ neighbors on the left.
\end{theorem}
\cref{MainThm-Informal} immediately follows from left expansion shown by \cite{kalev_unbalanced_2022} and \cref{thm: right sided lossless expansion- Informal}. For the rest of this section, we focus on proving \cref{thm: right sided lossless expansion- Informal} that relies on the following two key lemmas.
\begin{lemma}[Right regularity]\label{lem: overview right regular}
    The KT graph $G$ is right-regular and has right-degree $d_R=q^{n-(s+1)}$.
\end{lemma}

\begin{lemma}[Number of common left neighbors]\label{lem: overview overlap neighbor}
    For any pair of right-vertices $w_1, w_2\in\F_q^{s+2}$ such that $w_1 = (y_1, z_1), w_2 = (y_2, z_2)$ where $y_1\ne y_2\in\F_q$ and $z_1, z_2\in\F_q^{s+1}$, we have $\abs{\GammaR{(y_1,z_1)}\cap \GammaR{(y_2,z_2)}} = q^{n - (2s+2)}$ if $n \ge 2s+2$ and $\abs{\GammaR{(y_1,z_1)}\cap \GammaR{(y_2,z_2)}} \le 1$ if $n\le 2s+2$. 
\end{lemma}
\cref{thm: right sided lossless expansion- Informal} then follows by an application of the inclusion-exclusion principle---subtracting the maximum number of common neighbors between any pair of vertices in $T$ from the total number of edges leaving $T$---we get the required lower bound on the size of $T$'s left neighborhood.
\par We now discuss the proof techniques for showing \cref{lem: overview right regular} and \cref{lem: overview overlap neighbor}. We start by making a simple but useful observation on the structure of the KT graph $G$.
\begin{observation}\label{obs: same seed disjoint neighborhood}
    Fix $w=(y,z_0, \cdots, z_s) \in R$ and let $f\in L$ be any left-neighbor of $w$. Then it must be the case that $w$ is the $y$'th neighbor of $f$. Now for any $w'\in R$ such that $w'= (y, z_0', \cdots, z_s')$, it holds that $f \notin \GammaR(w')$. This is saying that any pair of right vertices $(w, w')$ that come from the same seed \footnote{We sometimes refer to $y \in \F_q$ as the ``seed'', like in the condensers literature.} must have disjoint left neighborhoods.
\end{observation}
Central to our analysis of the right degree and the number of common left neighbors are the following linear maps.
\begin{definition}
For $y\in \F_q$, define the map $\seed_y(f): \F_q^n \to \F_q^{s+1}$ as follows: Interpret $f\in \F_q[X]$ as a degree $\le n-1$ polynomial and map it to $(f^{(0)}(y), \dots, f^{(s)}(y))$ where $f^{(i)}$ is the $i$'th iterative derivative of $f$.
\end{definition}
We note that $\seed_y$ is a $\F_q$-linear map, for any $y \in \F_q$.
\begin{definition} 
For $y_1, y_2\in \F_q$, $y_1 \neq y_2$, define the map $\seed_{y_1,y_2}(f): \F_q^n \to \F_q^{2(s+1)}$ as the concatenation of the respective linear maps, that is, $\seed_{y_1,y_2}(f) =  (\seed_{y_1}(f) , \seed_{y_2}(f))$.
\end{definition}

Proving the above lemmas (about the KT graph) now boils down to analyzing both $\seed_y$ and $\seed_{y_1,y_2}$ for all $y, y_1, y_2 \in \F_q$.
\begin{enumerate}
    \item We show that $\psi_y$ for $y\in \F_q$ is surjective, which along with \cref{obs: same seed disjoint neighborhood} implies \cref{lem: overview right regular}: For any $w = (y, z_0, \cdots, z_s) \in R$, the set of its left neighbors is $\{f\in L \mathrel{|} (y, \seed_y(f)) = w\} = \seed^{-1}_y(z_0, \cdots, z_s)$. Therefore, the right degree $D_R = \abs{\seed^{-1}_y(z_0, \cdots, z_s)} = q^n/q^{s+1} = q^{n-(s+1)}$.
    \item We show $\seed_{y_1,y_2}$ is surjective when $n \geq 2s+2$ and injective $n \leq 2s+2$, which implies \cref{lem: overview overlap neighbor}: Similar to above, let $w_1= (y_1, z_1) \in R$ and $w_2 = (y_2, z_2) \in R$, $y_1 \neq y_2$, be any pair of right vertices from \textit{different} seeds. We extend \cref{obs: same seed disjoint neighborhood} to see that the number of $f \in L$ such that $(y_1, \seed_{y_1}(f)) = w_1$ and $(y_2, \seed_{y_2}(f)) = w_2$ is exactly $\abs{\seed_{y_1, y_2}^{-1}(z_1, z_2)}$. When $n\ge 2s+2$, this map is surjective, and the number of left neighbors shared by $w_1$ and $w_2$ is $q^{n - (2s+2)}$. When $n\le 2s+2$, this map is injective and the number of shared neighbors is at most $1$.
\end{enumerate}
To conclude the proof, we carry out Hermite interpolation (see \cref{lem: hermite general}) by applying the Chinese remainder theorem to show the surjectivity and injectivity of the maps $\seed_{y},\seed_{y_1}$ and $\seed_{y_2}$. 

\subsubsection{Tightness of right expansion}

We will show that our parameters from \cref{thm: right sided lossless expansion- Informal} are tight. Throughout this section, we let $P_d\subseteq\F_q[x]$ be the polynomials of degree exactly $d$ and $P_{<d}\subseteq\F_q[x]$ be the vector space over $\F_q$ of polynomials of degree strictly less than $d$. First, notice that when $|R| > |L| / |R|$, the right expansion is optimal by \cref{rmk:tightness when n>2s+2}. Hence, we only focus on the case when $s+1< n< 2s+2$ and show the following:
\begin{theorem}[Informal version of \cref{thm:epsilon and K tradeoff tight when n<2s+2}]
\label{thm-proof overview: tightness}
For $s+1 <n< 2s+2$ and any $0<\delta \le 2$, there exists $T\subseteq R$ such that $\abs{T} = \delta q^{n-s-1}$ and $\abs{\GammaR(T)}=(1-\varepsilon) D_R\abs{T}$ with $\varepsilon = \delta / 4$.
\end{theorem} 

For fixed $y_1\ne y_2\in \F_q$, we will construct $T = T_1 \sqcup T_2$ satisfying the following:
\begin{itemize}
    \item 
    $|T_1| + |T_2| = (\delta/2) q^{n-s-1}$.

    \item 
    The first coordinate of every element of $T_1$ is always $y_1$ and of $T_2$ is always $y_2$.

    \item
    For all $t_1\in T_1$ and $t_2\in T_2$, there exists a degree less than $n$ polynomial $f$ such that $\GammaL(f, y_1) = t_1$ and $\GammaL(f, y_2) = t_2$. 
\end{itemize}

We then observe that $\abs{\GammaR(T)} = \abs{\GammaR(T_1)} + \abs{\GammaR(T_2)} - \abs{T_1}\cdot \abs{T_2}$. Since the right degree of $G$ is $q^{n-s-1}$, we have that $\abs{\GammaR(T_1)} = \abs{\GammaR(T_2)} = (\delta/2) q^{2n-2s-2}$. Hence, we compute that $\abs{\GammaR(T)} = (\delta - \delta^2/4) q^{2n-2s-2} = (1-\delta/4) D_R\abs{T}$, proving \cref{thm-proof overview: tightness}.

To construct such $T_1, T_2$, it suffices to construct $S_1, S_2\subset \F_q^{s+1}$ with $|S_1| = |S_2| = K/2$ such that $S_1\times S_2\subset \psi_{y_1, y_2}(P_{<n})$ (see \cref{cor:lem:S1 and S2 construction for psi} for a formal claim). Indeed, we can let $T_1 = (y_1, S_1)$ and $T_2 = (y_2, S_2)$ and check that $T_1$ and $T_2$ have the desired properties.

Before we show how to construct such $S_1$ and $S_2$, we will need to introduce a few algebraic objects. Let $g_1(x) = (x - y_1)^{s+1}, g_2(x) = (x - y_2)^{s+1}$, and let $\varphi: \F_q[x] \to \F_q[x] /g_1\times \F_q[x] / g_2$ as $\varphi(f) = (f \mod g_1 , f \mod g_2)$. 

We then prove a structural result regarding $\varphi$:
\begin{lemma}[Informal version of \cref{lem:Structure of image of phi}]\label{lemma-proof-overview: structure of image of phi}
For $s+1< d< 2s+2$, we have 
$$\varphi(P_{<d}) = \bigcup_{h\in P_{<d-(s+1)}} \{(f, f + \sigma(h)) \mid f\in P_{\le s}\}.$$
\end{lemma}
Here, $\sigma: P_{\le s}\to P_{\le s}$ is a specially chosen injective homomorphism that we define later.

To construct such $S_1$ and $S_2$, we construct $R_1\subset \F_q[x]/g_1, R_2\subset F_q[x]/g_2$ such that $R_1\times R_2\subset \varphi(P_{<n})$ (see \cref{lem:S1 and S2 construction for phi} for formal claim). Once we have such $R_1, R_2$, we let $S_1 = \psi_{y_1}(R_1)$ and $S_2 = \psi_{y_2}(R_2)$. We then check that if $f\in P_{<n}$ is such that $\varphi(f) = (r_1, r_2)$, then $f(y_1) =  r_1(y_1) = s_1$ and $f(y_2) = r_2(y_2) = s_2$, showing that $S_1\times S_2\subset\psi_{y_1, y_2}(P_{<d})$, as desired.

We now construct such $R_1$ and $R_2$. Let $R_1$ and $R_2$ be arbitrary size $K/2$ subsets of $\sigma(P_{<n-(s+1)})$.
Note that since $\sigma$ is injective, $\abs{\sigma(P_{<n-(s+1)})}= q^{n-s-1} \ge K/2$ and so we can indeed pick such $R_1$ and $R_2$.
We claim that for any $r_1\in R_1, r_2\in R_2$, it holds that $(r_1, r_2)\in \varphi(P_{< n})$. Using \cref{lemma-proof-overview: structure of image of phi}, it suffices to show that $(r_1, r_2) = (f, f + \sigma(h))$ for some $f\in P_{<n}$ and $h\in P_{n-(s+1)}$. Since $R_1,R_2\subseteq\sigma(P_{<n-(s+1)})$, we know that $r_1 = \sigma(h_1)$ and $r_2 = \sigma(h_2)$ for some $h_1, h_2\in P_{<n-(s+1)}$.
Consequently,
\[
(r_1, r_2) 
= (\sigma(h_1), \sigma(h_2)) 
= (\sigma(h_1), \sigma(h_1) + \sigma(h_2-h_1))
\]
as desired. In the last line, we used the fact that $\sigma$ is a homomorphism and that $h_1, h_2\in P_{< n-(s+1)}$.

We finally define $\sigma$:
\begin{definition}
Recall that $\sigma:P_{<n-(s+1)}\to P_{\le s}$.
To compute $\sigma(h_1)$, we let $f\in P_{<n}$ be any polynomial such that $f = h_1g_1 + r_1$ where $r_1\in P_{\le s}$ is the remainder of $f$ modulo $g_1$. We then write $f = h_2g_2 + r_2$ where $r_2\in P_{\le s}$ is the remainder of $f$ modulo $g_2$. The output of $\sigma(h_1)$ is $r_2 - r_1$.
\end{definition}

At first glance, it seems unclear whether $\sigma$ is even a well defined function. To help show this, we define $\rho: P_{<n-(s+1)}\to P_{<n-(s+1)}$ such that $\rho(h_1) = h_2$ where $h_1, h_2$ are defined as above. We first show that $\rho$ is an isomorphism. To do this, we observe that $h_1g_1 - h_2g_2 = r_2 - r_1$ has degree at most $s$, and so it must be that $h_1g_1$ and $h_2g_2$ agree on all coefficients corresponding to degree $\ge s+1$. We compare these coefficients and, using linear algebra, show that we can indeed obtain a unique $h_2$ from $h_1$, showing it is indeed a function. This argument in fact directly shows that $\rho$ is an isomorphism.

Once we have that $\rho$ is an isomorphism, we can then define $\sigma(h_1) = h_1g_1 - \rho(h_2)g_2$. Using this, it easily follows that $\sigma$ is a homomorphism. From this we obtain \cref{lemma-proof-overview: structure of image of phi}. By a further counting argument, we can show that $\sigma$ is injective. For details, see \cref{subsec:structure of image of phi}.

\subsection{Non-bipartite lossless expander}\label{subsec: proof overview non-bipartite}

We show that the bipartite half of the KT graph (from the previous section) yields a non-bipartite regular lossless expander.
The bipartite half is an operation of bipartite graphs that transforms them into a non-bipartite graph, and is defined as follows: given a bipartite graph $G = (L\sqcup R, E)$, its bipartite half $G^2[L]$ is a graph with vertex set $L$ where there is an edge $(u, v) \in G^2[L]$ iff $u$ and $v$ share a common neighbor in $G$.

One nuance of the bipartite half is that applying it to a biregular bipartite graph does not necessarily mean that the bipartite half will be regular itself (although the graph we obtain from \cite{kalev_unbalanced_2022} will indeed be regular). Thus, we must define what it means for a graph to be lossless in this non-regular setting. A natural definition just involves summing the total number of neighbors of a set.
\begin{definition}\label{def:irregular lossless}
    An irregular graph $G=(V,E)$ is a \emph{$(K,\varepsilon)$-lossless expander}\footnote{We abuse notation between the regular and irregular cases of graphs since this definition of lossless expansion for an irregular graph captures our previous definition of lossless expansion for regular graphs.} if for any set $S\subseteq V$ of size at most $K$ we have that $\abs{\Gamma(S)}\geq (1-\varepsilon)\sum_{v\in S}d(v)$ where $d(v)$ represents the degree of vertex $v$. 
\end{definition}
A stronger notion of lossless expansion is with respect to the highest degree of a node present in a graph. 
\begin{definition}\label{def:max-degree irregular lossless}
    An irregular graph $G=(V,E)$ is a \emph{max-degree $(K,\varepsilon)$-lossless expander} if for any set $S\subseteq V$ of size at most $K$ we have that $\abs{\Gamma(S)}\geq (1-\varepsilon)D\abs{S}$ where $D=\max_{v\in V}d(v)$, the maximum degree of any vertex in $G$.
\end{definition}

Using this definition, our main observation is that the bipartite half of any two-sided lossless bipartite expander yields a non-bipartite, max-degree lossless expander.

\begin{lemma}[\cref{lem:Two-sided lossless expander to lossless expander} restated]\label{lem:overview-Two-sided lossless expander to lossless expander}
Let $G=(L\sqcup R,E)$ be a $(D_L,D_R)$-regular $(K_L,A_L,K_R,A_R)$-two-sided lossless expander. Then $G^2[L]$ is a max-degree $(K,A)$-expander where each node has a degree in $[D_LA_R,D_LD_R]$, and with $K=\min(K_L,K_R/D_L)$ and $A=A_LA_R$.
\end{lemma}
The proof of this lemma essentially follows from expanding twice in the underlying two-sided expander $G$. Since we force our initial set to be at most $K_L$ and $K_R/D_L$, we are guaranteed that we can use the left-to-right expansion of $G$ and then additionally the right-to-left expansion of $G$, where at each step we expand by $A_L$ and $A_R$, respectively. 

Finally, we use the bipartite two-sided lossless expander from \cref{thm: informal bipartite pos and impos} as the base graph in \cref{lem:overview-Two-sided lossless expander to lossless expander} to obtain \cref{thm: informal non-bipartite pos and impos}.
Luckily, if we use the KT graph as our bipartite two-sided lossless expander, then the resultant graph obtained from taking the bipartite half is indeed regular (see \cref{lem:KT bipartite half is regular} for a proof).

\begin{theorem}[Informal version of \cref{thm:normal lossless expander}]\label{thm:overview-normal lossless expander}
  For infinitely many $N$ and all constant $0 < \delta < 0.99$, there exists an explicit regular $(K, \eps = 0.01)$ lossless expander $G = (V, E)$ where $|V| = N$, the degree is $D$ where $N^{1 - 1.01\delta}\le D \le N^{1 - 1.01\delta + o(1)}$ and $K = \min\left(N^{\delta}, N^{1 - 1.01\delta - o(1)}\right)$. Moreover, $G$ is endowed with a free group action from $\F_q$ where $q=\poly(\log N)$ if one vertex is removed.
\end{theorem}

In this setting, $A = A_LA_R \approx 0.99 D_LD_R$, implying $G^2[L]$ is indeed a max-degree lossless expander. Additionally, because the vertices in the bipartite half of the KT graph are elements of $\F_q^n$, we get a free group action from $\F_q$ on them by scalar multiplication. One needs to be careful here since $G^2[L]$ contains the zero polynomial vertex; we remove this vertex and observe that removing one vertex still preserves the expansion properties.

\paragraph{Organization} We use \cref{sec:prelim} to introduce necessary preliminaries. Then in \cref{subsec:main theorem} we show how our main theorem is proved assuming right regularity and knowing the overlap between two neighborhoods of right vertices. These facts are then proved in \cref{sec: hermite}. In \cref{sec: plug in param}, we plug in parameters to get our two-sided lossless expander. In \cref{sec:tightness} we prove tightness of our right-to-left expansion analysis of the KT graph. Finally, in \cref{sec:non-bipartite lossless expander} we show how the bipartite half of the KT graph is a non-bipartite lossless expander with a free group action.

We prove that our constructions are explicit in \cref{sec:expliciteness}, and discuss why our techniques do not work for the \cite{guruswami_unbalanced_2009} graph in \cref{sec:GUV not right regular}.

\section{Preliminaries}\label{sec:prelim}

\subsection{Notation}

For a function $f\in \F_q[X]$, we we use $f^{(j)}$ to denote the $j$'th iterated derivative of $f$.
We will often use the notation $b_i$ for $i\in \N$ to refer to the polynomial $x^i\in \F_q[x]$ and we will often use the fact that $(b_0, \dots, b_n)$ form a basis for the polynomials of degree at most $n$.
For a $(d_L, d_R)$-biregular bipartite graph $G = (L\sqcup R)$, we use $\GammaL: L\times [D_L]\to R$ to be the function that maps vertices in $L$ to their neighbors in $R$ as given by $G$; we use $\GammaR: R\times [D_R]\to L$ to be the function that maps vertices in $R$ to their neighbors in $L$ as given by $G$. Often, we will define graph $G$ by only defining the associated $\GammaL$. When clear from context, we sometimes abuse notation and use $\GammaL(w)$ to denote the right neighborhood of $w\in L$, and similarly $\GammaR(w)$ for the left neighborhood of $w\in R$.

\subsection{Lossless expansion}

Throughout this paper, we will be focusing on the notion of vertex expansion as opposed to other definitions (e.g., edge, spectral) of expansion. Defining vertex expansion of a regular graph is straightforward.
\begin{definition}\label{def:prelim-expander}
    A $D$-regular graph $G=(V,E)$ is a \emph{$(K,A)$-expander} if for all $S\subseteq V$ such that $\abs{S}\leq K$ we have that $\abs{\Gamma(S)}\geq A\abs{S}$. If $A=1-\varepsilon$, then we say that $G$ is a \emph{$(K,\varepsilon)$-lossless expander}.
\end{definition}

For biregular bipartite graphs, we must consider the degree of each side to define expansion.
\begin{definition}\label{def:prelim-two-sided lossless expander}
    A $(D_L,D_R)$-biregular graph $G=(L\sqcup R,E)$ is a \emph{$(K_L,A_L,K_R,A_R)$-two-sided expander} if for all $S\subseteq L$ of size at most $K_L$ we have $\abs{\GammaL(S)}\geq A_L\abs{S}$ and for all $S\subseteq R$ of size at most $K_R$ we have $\abs{\GammaR(S)}\geq A_R\abs{S}$. If $A_L=1-\varepsilon_L$ and $A_R=1-\varepsilon_R$, then we call $G$ a \emph{$(K_L,\varepsilon_L,K_R,\varepsilon_R)$-lossless two-sided expander}.
\end{definition}

For irregular graphs, we can generalize \cref{def:prelim-expander} in two ways. The first way is considering expansion with respect to the maximum number of neighbors of a set.

\begin{definition}\label{def:prelim-irregular lossless}
    An irregular graph $G=(V,E)$ is a \emph{$(K,\varepsilon)$-lossless expander} (where we abuse the word ``expander'' for both regular and irregular graphs) if for any set $S\subseteq V$ of size at most $K$ we have that $\abs{\Gamma(S)}\geq (1-\varepsilon)\sum_{v\in S}d(v)$ where $d(v)$ represents the degree of vertex $v$. 
\end{definition}
The second, stronger notion of lossless expansion is with respect to the highest degree of a node present in a graph. 
\begin{definition}\label{def:prelim-max-degree irregular lossless}
    An irregular graph $G=(V,E)$ is a \emph{max-degree $(K,\varepsilon)$-lossless expander} if for any set $S\subseteq V$ of size at most $K$ we have that $\abs{\Gamma(S)}\geq (1-\varepsilon)D\abs{S}$ where $D=\max_{v\in V}d(v)$, the maximum degree of any vertex in $G$.
\end{definition}

\subsection{The KT graph}

Throughout the paper, we will use construction of bipartite (left) lossless expanders from \cite{kalev_unbalanced_2022} based on multiplicity codes from \cite{kopparty_high-rate_2014}.
We will often refer to this graph `the KT graph':
\begin{definition}[The KT graph]\label{def: KTGraph full}
Let $q, n, s\in \N$ be such that $q$ is a prime power, characteristic of the finite field $\F_q \ge n$ and $s \le n / 2$. Define $G = (L \sqcup R, E)$ where $L = \F_q^n, R = \F_q^{s+2}$. The left degree is $q$ and for any $f\in \F_q^n$ and $y\in \F_q$, the $y$'th neighbor of $f$ is defined as follows: Identify $f$ as member of $\F_q[X]$ with degree of $f$ at most $n-1$ ; then, the neighbor $\GammaL{(f,y)}$ will be $(y, f^{(0)}(y), \dots, f^{(s)}(y))$ where $f^{(j)}$ is the $j$'th iterative derivative of $f$.
\end{definition}

\begin{remark}
In the paper \cite{kalev_unbalanced_2022}, the final lossless expander graph construction slightly differs from ours. While they do construct the KT-graph $G$ defined as above and show it has great (left) expanding properties, the final (left) lossless expander graph actually is defined as $H = (L \sqcup R, E)$ where $L = 2^n, R = \F_q^{s+2}$ and the left degree is $q$. $H$ is constructed by considering the subgraph of $G$ induced by vertices on the left side corresponding to $\{0, 1\}^n$. For us, the final two-sided lossless expander graph will be $G$ itself. This is why, our two-sided lossless expander graph has slightly worse parameters (worse constants) compared to the left lossless expander graph from \cite{kalev_unbalanced_2022}.
\end{remark}

\subsection{A useful inequality}
We will use the following inequality based on an application of the Cauchy-Schwarz inequality:
\begin{claim}\label{claim: pairwise sum inequality}
Fix $n\in \N, S\in \R$.
Let $x = (x_1, \dots, x_n)\in \R^n$ be such that $\sum_{1\le i\le n} x_i = S$.
Then,
\[
\sum_{1\le i < j \le n} x_i x_j \le \frac{(n-1)S^2}{2n}
\]
\end{claim}
\begin{proof}
Recall the Cauchy-Schwarz inequality:
$\left(\sum_{1\le i\le n} a_i b_i\right)^2 \le \left(\sum_{1\le i\le n} a_i^2\right)\left(\sum_{1\le i\le n} b_i^2\right)$.
We apply this with $a_1 = x_1, \dots, a_n = x_n$ and $b_1 = b_2 = \dots = b_n = 1$ to infer that
\[
S^2 \le \left(\sum_{1\le i\le n} x_i^2\right)\cdot n = \left(S^2 - 2\sum_{1\le i < j \le n} x_ix_j\right)\cdot n
\]
Rearranging, we infer that 
\[
\sum_{1\le i < j \le n} x_i x_j \le \frac{(n-1)S^2}{2n}
\]
as desired.
\end{proof}

\subsection{Free group actions on graphs}\label{subsec:prelim-free group action}
Here we recall basic notions about group actions on graphs. First, we define an abstract group notion.
\begin{definition}\label{def:group action}
    Let $G$ be a group and $X$ a set. A \emph{group action} $\cdot:G\times X\to X$ (where we write the $\cdot$ in infix notation) is a function that has the following two properties:
    \begin{enumerate}
        \item Identity: The identity element $1_G$ of $G$ always acts trivially as $1_G\cdot x=x$ for any $x\in X$.
        \item Compatibility: The group action and multiplication of $G$ are compatible. That is, for any $g,h\in G$ and $x\in X$ we have $(gh)\cdot x=g\cdot(h\cdot x)$ where $gh$ is the product of $g$ and $h$ in $G$.
    \end{enumerate}
\end{definition}

Next, we recall another abstract notion about group actions.
\begin{definition}\label{def:free group action}
    We say that a group action of $G$ on $X$ is \emph{free} if $g\cdot x=x$ for some $x\in X$ implies that $g=1_G$. 
\end{definition}

Finally, we consider what it means for a graph to be invariant with respect to a group action.
\begin{definition}\label{def:G-invariant graph}
    Let $G$ be a group and $H=(V,E)$ a graph with a group action from $G$. We say that $H$ is \emph{$G$-invariant} if for all $(v,w)\in E$ and $g\in G$ we have that $(g\cdot v,g\cdot w)\in E$. 
\end{definition}

\section{An Explicit Two-sided Lossless Expander}
In this section, we first describe how to prove our main theorem using right regularity and the size of the overlap in neighborhoods between any two right vertices. Then we prove these two facts in \cref{sec: hermite}.
    \subsection{Main theorem}\label{subsec:main theorem}
    Putting together all of our results with the left-to-right expansion of \cite{kalev_unbalanced_2022} yields our main theorem.
        \begin{theorem}\label{MainThm}
            For all finite fields $\F_q$ and $n, s\in \N$ with $15\le (s+1) < n < char(\F_q)$, there exists an explicit bipartite graph $G = (L\sqcup R, E)$ with $L = \F_q^n, R = \F_q^{s+2}$, left degree equal to $q$ and right degree $q^{n - (s+1)}$ such that $G$ is a
            two-sided $(K_L, A_L, K_R, A_R)$ expander
            with $A_L = q - \frac{n(s+2)}{2}\cdot \left(qK_L\right)^{1 / (s + 2)}$ and $A_R = \left(1 - \frac{K_R}{q^{\min(s+2, n-s)}}\cdot \frac{q-1}{2}\right)q^{n-(s+1)}$.
        \end{theorem}
        \begin{proof}
            The left-to-right expansion follows from Theorem 3 from \cite{kalev_unbalanced_2022}. The right-to-left expansion follows from \cref{thm: right sided lossless expansion} below. The explicitness of $G$ follows from \cref{claim: explicitness of graph}.
        \end{proof}
        
        Our main achievement is showing the right-to-left expansion of the KT graph in \cref{thm: right sided lossless expansion} below.
        \begin{theorem}\label{thm: right sided lossless expansion}
            If $n \ge s+1$, then the KT graph $G$ in \cref{def: KTGraph} is a right $(K_{max},\eps)$-lossless expander for $K_{max}=\delta q^{s+1}$ and $\varepsilon = \frac{\delta(q-1)}{2q}\cdot q^{\max(2s+2-n, 0)}$ where $0 < \delta < 1$ is arbitrary.
        \end{theorem}

        
        We prove \cref{thm: right sided lossless expansion} via the following properties of $G$:
        \begin{lemma}\label{thm: G right regular and right degree}
            When $n \ge s+1$, $G$ is right-regular and the right degree is $q^{n - (s+1)}$.
        \end{lemma}
        
        \begin{lemma}\label{thm:two vertex overlap}
            For any pair of right-vertices $w_1, w_2$ such that $w_1 = (y_1, z_1), w_2 = (y_2, z_2)\in\F_q^{s+2}$ where $y_1\ne y_2\in\F_q$ and $z_1, z_2\in\F_q^{s+1}$, we have
            \[
            \abs{\GammaR{(y_1,z_1)}\cap \GammaR{(y_2,z_2)}}\le
            \begin{cases}
                q^{n - (2s+2)} & n \ge 2s+2\\
                1            & n \le 2s+2
            \end{cases}
            \]
        \end{lemma} 

        We will prove both these lemmas in \cref{sec: hermite}.
        

        With the exact right-regularity of $G$ and the number of common left-neighbors shared by any pair of right-vertices generated by different seeds, we are ready to prove \cref{thm: right sided lossless expansion}.
        \begin{proof}[Proof of \cref{thm: right sided lossless expansion}]
            Our goal is to show that any right subset $T\subseteq \F_q^{s+2}$ of size at most $\delta q^{s+1}$ has a neighborhood of size at least $(1-\varepsilon)q^{n-(s+1)}\abs{T}$ on the left. 
    
            To do this, we consider $T$ as the disjoint union $T=\bigsqcup_{y\in\F_q}T_y$ of buckets $T_y=\{(y,\alpha):\alpha\in\F_q^{s+1}\}$ where $\abs{T_y} = t_y = \delta_y q^{s+1}$.
            Let $\delta = \sum_{y\in \F_q} \delta_y$.
            So, $\abs{T} = \delta q^{s+1}$.
            By \cref{thm: G right regular and right degree}, the number of edges leaving $T$ is $\abs{T}\cdot q^{n-(s+1)} = \delta q^n$.

            We now consider cases on whether $n \ge 2s+2$ or not:
            \begin{enumerate}
                \item[Case 1.] $n\ge 2s+2$.\\
                In this case, $\eps = \frac{\delta (q-1)}{2q}$.
                By \cref{thm:two vertex overlap}, the maximum number of double-counted left vertices is
                \[
                \sum_{\substack{i,j\in[q]\\i<j}} t_it_jq^{n-2(s+1)}
                = \sum_{\substack{i,j\in[q]\\i<j}}\delta_iq^{s+1}\cdot\delta_jq^{s+1}\cdot q^{n-2(s+1)}
                = q^n\sum_{\substack{i,j\in[q]\\i<j}}\delta_i\delta_j\le q^n\cdot \frac{q-1}{2q}\cdot \delta^2
                \]
                where for the last inequality, we used \cref{claim: pairwise sum inequality}.
                Applying one level of inclusion-exclusion reveals that
                \[
                    \abs{\GammaR{(T)}}\geq\delta q^n - q^n\cdot \frac{q-1}{2q}\cdot \delta^2
                    = \left(1 - \frac{\delta (q-1)}{2q}\right)\delta q^n
                    = \left(1 - \eps\right)q^{n-(s+1)}|T|
                \]
                where the last equality follows because $\eps =\frac{\delta(q-1)}{2q}$.
                
                \item[Case 2.] $2s+2\ge n\ge s+1$.\\
                In this case, $\eps = \frac{\delta (q-1)}{2q}\cdot q^{2s+2-n}$.
                By \cref{thm:two vertex overlap}, the maximum number of double-counted left vertices is
                \[
                \sum_{\substack{i,j\in[q]\\i<j}} t_it_j
                = \sum_{\substack{i,j\in[q]\\i<j}}\delta_iq^{s+1}\cdot\delta_jq^{s+1}
                = q^{2s+2}\sum_{\substack{i,j\in[q]\\i<j}}\delta_i\delta_j\le q^{2s+2}\cdot \frac{q-1}{2q}\cdot \delta^2
                \]
                where for the last inequality, we used \cref{claim: pairwise sum inequality}.
                We again apply one level of inclusion-exclusion to conclude that
                \[
                    \abs{\GammaR{(T)}}\geq\delta q^n - q^{2s+2}\cdot \frac{q-1}{2q}\cdot \delta^2
                    = \left(1 - \frac{\delta (q-1)}{2q}\cdot q^{2s+2-n}\right)\delta q^n
                    = \left(1 - \eps\right)q^{n-(s+1)}|T|
                \]
                where the last equality follows because $\eps =\frac{\delta(q-1)}{2q}\cdot q^{2s+2-n}$.
            \end{enumerate}
        \end{proof}

    \subsection{Right regularity and bounding common neighbors: Hermite interpolation}\label{sec: hermite}
        In this section, we show the $(q^{n-(s+1)})$-right-regularity of the KT graph G, and bound the number of common left neighbors shared by any pair of right vertices with different seeds. Both tasks are essentially a question of Hermite interpolation—we wish to find polynomials $f \in \F_q[Y]$ of degree at most $n-1$ such that when evaluating at some point $y\in \F_q$, the function value $f(y)$ and its first $s$ derivatives $(f^{(0)}(y), f^{(1)}(y), \cdots , f^{(s)}(y))$ match the values given by the right vertices. 

        \begin{lemma}[Hermite interpolation] \label{lem: hermite general} 
            Let $y_1, \cdots, y_k \in \F_q$ be distinct, and for $i \in [k]$, let $z_{i,0}, \cdots , z_{i,s} \in \F_q$. Then there exists a unique polynomial $f\in \F_q[Y]$ with degree at most $k(s+1)$ such that $f^{(j)}(y_i) = z_{i,j}$ for $i\in [k]$ and $j\in \{0\}\cup [s]$.
        \end{lemma}
        \begin{proof}
            Consider the following $k$ congruences,
            \[
                f(Y) \equiv f_1(Y) \mod (Y-y_1)^{s+1} ,\ 
                  \cdots, \
                 f(Y) \equiv f_k(Y) \mod (Y-y_k)^{s+1}
            \]
            Any polynomial $f$ that satisfies the above $k$ congruences must also satisfy $f^{(j)}(y_i) = z_{i,j}$ for $i\in [k]$ and $j\in \{0\}\cup [s]$— thus solving the interpolation—because $f_i(Y)$ is the order $s$ Taylor polynomial of $f$ at $y_i$.

            We conclude the proof by applying the Chinese remainder theorem for univariate polynomials which asserts that there exists a \textit{unique} polynomial $f \in \F_q[Y]$ of degree at most $k(s+1)$ that satisfies the above congruences. 
            
        \end{proof}

        Recall the maps $\psi_{y_1} : \F_q^n \rightarrow \F_q^{s+1}$, and $\psi_{y_1,y_2}: \F_q^n \rightarrow \F_q^{2s+2}$ where $\seed_y(f)=(f^{(0)}(y), \dots, f^{(s)}(y))$ and $\seed_{y_1, y_2}=  (\seed_{y_1}(f), \seed_{y_2}(f))$.
        We will use the following fact regarding linearity of derivatives:
        \begin{fact}[\cite{ritt1950differential}]\label{fact:derivative is linear}
            For all $\alpha, \beta\in \F_q$, $f, g\in \F_q[X]$ and $j\ge 0$, it holds that $(\alpha f + \beta g)^{(j)} = \alpha f^{(j)} + \beta g^{(j)}$.
        \end{fact}
        
        From this fact, we directly obtain that $\seed_y$ is a linear map:
        
        \begin{corollary}\label{cor: seedy is fq linear}
        For all $y\in \F_q$, $\seed_y$ is an $\F_q$-linear map.
        \end{corollary}

        Now we prove the right-regularity of $G$ (\cref{thm: G right regular and right degree}) and bound the number of overlapping left neighbors (\cref{thm:two vertex overlap}) as special cases of \cref{lem: hermite general}.
        
        \begin{proof}[Proof of \cref{thm: G right regular and right degree}]
            Let $w_1\in \F_q^{s+2}$ where $w_1 = (y_1, z_{1,0}, \cdots, z_{1,s})$.
            Take $k=1$ in \cref{lem: hermite general}, we get that there exists a unique polynomial $f$ of degree at most $s+1$ such that $f^{(j)}(y_1) = z_{1,j}$ for $j \in \{0\}\cup [s]$. This means the linear map $\psi_{y_1} : \F_q^{n} \rightarrow \F_q^{s+1}$ is surjective. Therefore, the number of left neighbors of any right vertex $w_1$ is exactly $|\psi^{-1}_{y_1} (z_{1,0}, \cdots, z_{1,s})| = q^{n-(s+1)}$.

        \end{proof}

        \begin{proof}[Proof of \cref{thm:two vertex overlap}]
           Let $w_1, w_2 \in \F_q^{s+2}$ where $w_1 = (y_1, z_{1,0}, \cdots, z_{1,s})$, $w_2 = (y_2, z_{2,0}, \cdots, z_{2,s})$, $y_1 \neq y_2$.
           Take $k=2$ in \cref{lem: hermite general}, we get that there exists a unique polynomial $f$ of degree at most $2(s+1)$ such that $f^{(j)}(y_i) = z_{i,j}$ for $i \in \{1,2\}, j \in \{0\}\cup [s]$. This means,
           \begin{itemize}
               \item When $n>2s+2$, $\psi_{y_1,y_2}: \F_q^n \rightarrow \F_q^{2s+2}$ is surjective, the number of common neighbors shared by $w_1, w_2$ is exactly $|\psi^{-1}_{y_1,y_2} (z_{1,0}, \cdots, z_{1,s},z_{2,0}, \cdots, z_{2,s})| = q^{n-2(s+1)}$.
               \item When $n\leq 2s+2$, $\psi_{y_1,y_2}: \F_q^n \rightarrow \F_q^{2s+2}$ is injective, $w_1, w_2$ share \textit{at most} 1 common left neighbor.
           \end{itemize}

        \end{proof}

\subsection{Plugging in the Parameters}\label{sec: plug in param}

We record our main results regarding two sided lossless expanders:

\begin{theorem}[Formal version of \cref{thm: informal bipartite pos and impos}]\label{thm: main simplified theorem}
For infinitely many $N$ and all $0 < \delta < 0.99$, there exists an explicit biregular two-sided $(K_{L}, \eps_L = 0.01, K_{R}, \eps_{R} = 0.01)$ lossless expander $\GammaL: [N]\times [D_L]\to [M]$ where $D_L \le O(\log^{204}(N))$, $N^{1.01\cdot \delta - o(1)}\le M \le D_L\cdot N^{1.01\cdot \delta}$, $K_{L} = N^{\delta}$, $K_{R} = \frac{1}{50}\cdot (1 / D_L)\cdot \min(M, N / M)$.
\footnote{Our theorem statement doesn't have any additional constraint on $D_R$ since it can be uniquely inferred from $N, M, D_L$.}
\end{theorem}

These will follow from the following technical lemma:

\begin{lemma}\label{thm: instantiate two sided lossless expanders}
Let $\alpha, \eps_L, \eps_R \in (0, 1)$ and $K_{R}, n, k_L, q\in \N$ be such that $q$ is a prime number, $\frac{h^{1+\alpha}}{2} \le q\le h^{1+\alpha}$ where $h = (4nk_L/\eps_L)^{1/\alpha}$ and such that both $\frac{4}{k_L}\log(2n/\eps_L) \le \alpha$ and $k_L(1+\alpha)\le n$.
Then, there exists an explicit biregular $(K_{L}, \eps_L, K_{R}, \eps_{R})$ two-sided lossless expander $\GammaL: [N]\times [D_L]\to [M]$ where $N = q^n, K_{L} = q^{k_L}, K_{L}^{1+\alpha - 1 / \log(h)} \le  M \le D_L\cdot K_{L}^{1+\alpha}$, $D_L \le O(\log (N) \log(K_{L}) / \eps_L)^{1 + 1/\alpha + o(1)}$, $\frac{K_{R}}{q^{\min(s+2, n-s)}}\cdot \frac{q-1}{2} \le \eps_R$ where $s+2 = \lceil k_L / \log_q(h) \rceil$.
\end{lemma}

We will instantiate this lemma using simple parameters to obtain our main theorems:

\begin{proof}[Proof of \cref{thm: main simplified theorem}]
We plug in $\alpha = 0.01, \eps_L = 0.01, \eps_R = 0.01, k_L = \delta n$ in \cref{thm: instantiate two sided lossless expanders} to obtain the desired lossless expander.
\end{proof}

We finally prove our main technical lemma using two-sided expander from \cref{MainThm}:

\begin{proof}[Proof of \cref{thm: instantiate two sided lossless expanders}]
As $s + 2 = \lceil k_L / \log_q(h) \rceil$, we have that $h^{s+1}\le K_{L}\le h^{s+2}$.
Observe that 
\[
s+1 < \frac{k_L\log(q)}{\log(h)} \le k_L(1+\alpha) \le n
\]
So, we can apply \cref{MainThm} and infer that there exists a graph $\GammaL: \F_q^n\times \F_q\to \F_q^{s+2}$ that is a $(\le h^{s+2}, A_L)$ left expander and $(\le K_R, A_R)$ right expander where $A_L = q - \frac{n(s+2)}{2}\cdot (qh^{s+2})^{1/(s+2)}$ and $A_R = \left(1 - \frac{K_{R}}{q^{\min(s+2, n-s)}}\cdot \frac{q-1}{2}\right)q^{n-(s+1)}$.
Notice that as $K_{L}\le h^{s+2}$, $\GammaL$ is indeed a $(K_{L}, A_L)$ expander.

\begin{itemize}
\item 
We first bound the left degree $D_L$:
\begin{align*}
D_L 
= q 
& \le h^{1+\alpha} 
= (4nk_L / \eps_L)^{1 + 1/\alpha} 
= (4\log(N)\log(K_{L}) / \log^2(q)\eps_L)^{1 + 1/\alpha}\\
& = (4\log(N)\log(K_{L}) / \eps_L)^{1 + 1/\alpha}\cdot \log^{2+2/\alpha}(q)
\end{align*}
This implies that 
\[
D_L = q\le (4\log(N)\log(K_{L}) / \eps_L)^{1 + 1/\alpha}\left(\log\left(8\log(N)\log(K_{L}) / \eps_L\right)\right)^{2+2/\alpha}
\]
Then indeed, $D_L \le O(\log (N) \log(K_{L}) / \eps_L)^{1 + 1/\alpha + o(1)}$.

\item
We now bound the number of right vertices $M$:
\begin{align*}
M 
= q^{s+2}
\le q\cdot h^{(1+\alpha)(s+1)}
\le q\cdot K_{L}^{1+\alpha}
\end{align*}
Additionally,  
\begin{align*}
M 
= q^{s+2}
\ge q^{K_L\log(q) / \log(h)}
\ge q^{K_L((1+\alpha)(\log h) - 1) / \log(h)}
= q^{K_L(1+\alpha) - K_L / \log(h)}
\end{align*}

\item
We now show lossless expansion from the right side:
\[
A_R = \left(1 - \frac{K_{R}}{q^{s+2}}\cdot \frac{q-1}{2}\right)q^{n-(s+1)} \ge (1 - \eps_R)D_R
\]
where the last inequality follows because $\frac{K_{R}}{q^{\min(s+2, n-s)}}\cdot \frac{q-1}{2} \le \eps_R$.

\item
We finally show lossless expansion from the left side:
First, we note that $s+2\le 2k_L$. Indeed, $s+2 \le k_L\log_q(h) + 1 = k_L\frac{\log(h)}{\log(q)} + 1 \le k_L(1+\alpha) + 1 \le 2k_L$.
Then,
\begin{align*}
A_L
& = q - \frac{n(s+2)}{2}\cdot (qh^{s+2})^{1/(s+2)}\\
& =  q - \frac{n(s+2)h}{2}\cdot (q)^{1/(s+2)}\\
& \ge  q - nk_Lh\cdot (q)^{1/(s+2)} & \textrm{(since $s+2 \le 2k_L$)}\\
& =  q - \frac{\eps_L\cdot h^{\alpha}}{4}\cdot h\cdot (q)^{1/(s+2)} & \textrm{(since $nk_L = (\eps_L\cdot h^{\alpha}) / 4$)}\\
& =  q - \eps_L\cdot\frac{h^{1+\alpha}}{4}\cdot (q)^{1/(s+2)}\\
& \ge  q - \frac{\eps_L}{2}\cdot q \cdot (q)^{1/(s+2)} & \textrm{(since $h^{1+\alpha}/2 \le q$)}\\
& =  q\left(1 - \eps_L\cdot \frac{(q)^{1/(s+2)}}{2}\right)\\
& \ge  q(1 - \eps_L)\\
\end{align*}
The last inequality $(q)^{1/(s+2)} \le 2$ follows because we claim that $s+2 \ge \log(q)$. This suffices to prove the last inequality since then $(q)^{1/(s+2)} \le q^{1 / \log(q)} \le 2$.
We indeed compute that
\[ 
s+2 
\ge \frac{k_L}{\log_q(h)} 
\ge \frac{k_L((1+\alpha)\log(h)-1))}{\log(h)} 
\ge k_L 
\]
Moreover, as $\alpha \ge \frac{4}{k_L}\log(2n/\eps_L)$, we infer that 
$k_L \ge \frac{4}{\alpha}\cdot \log(2n/\eps_L)$.
Hence indeed,
\[
s+2
\ge \frac{4}{\alpha}\cdot \log(2n/\eps_L)
\ge \frac{2}{\alpha}\cdot \log(2nk_L/\eps_L)
\ge 2\log(h)
\ge (1+\alpha)\log(h)
\ge \log(q)
\]

\end{itemize}
\end{proof}

\section{Tightness of Our Construction}\label{sec:tightness}
In this section we show that the right-to-left expansion of \cref{thm: right sided lossless expansion} is tight. In particular, when $n< 2s+2$ \cref{thm: right sided lossless expansion} gives a trade-off between the expansion parameter $\varepsilon$ and the max size of expanding sets $K_{max}$. We show that this trade-off is tight up to constants, and thus fully characterize the behavior of the KT graph. Importantly, we show that in the balanced setting where $n=s+O(1)$, the KT graph is \emph{not} a two-sided lossless expander.

Recall that when $n\geq 2s+2$ we know that our result is tight as stated in \cref{rmk:tightness when n>2s+2}. Consequently, our main theorem in this section deals with the regime where $s+1 < n < 2s+2$. In this setting, \cref{thm: right sided lossless expansion} gives us that sets $S\subseteq R$ on the right of size at most $K_{max}=\delta q^{s+1}$ expand with parameter $\varepsilon=\frac{\delta(q-1)}{2q}\cdot q^{2s+2-n}$. Equivalently, it gives us that sets of size at most $K_{max} = \delta q^{n-s-1}$ expand with parameter $\varepsilon=\frac{\delta(q-1)}{2q}$. Our main theorem in this section upper bounds this expansion.

\begin{theorem}\label{thm:epsilon and K tradeoff tight when n<2s+2}
    When $s+1< n< 2s+2$, there exists a subset $S\subseteq R$ of the right vertices such that $\abs{S}= K_{max}=\delta q^{n-s-1}$ and $\abs{\GammaR(S)}=(1-\varepsilon) D_R\abs{S}$ with $\varepsilon=\frac{\delta}{4}$ where $\delta>0$. 
\end{theorem}
This means that our right-to-left expansion of $\cref{thm: right sided lossless expansion}$ is tight up to a constant factor of $1/2$.
The proof of our main theorem comes from the construction of two disjoint subsets of right vertices each in different buckets that have the maximum number of overlapping neighbors on the left. Recall that the $y$-th bucket $T_y$ is defined as $T_y=\{(y,\alpha):\alpha\in\F_q^{s+1}\}$

\begin{lemma}\label{lem:exist bad S1 and S2}
    Let $y_1,y_2\in\F_q$ be arbitrary such that $y_1\neq y_2$. Then there exist sets $S_1\subseteq T_{y_1}$ and $S_2\subseteq T_{y_2}$ such that $\abs{S_1}=\abs{S_2}= \frac{K_{max}}{2}$ and $\abs{\GammaR(S_1)\cap\GammaR(S_2)}=\abs{S_1}\cdot\abs{S_2}$.
\end{lemma}

Using this lemma, our main theorem is a result of straightforward computations.
\begin{proof}[Proof of \cref{thm:epsilon and K tradeoff tight when n<2s+2}]
    Take any $y_1,y_2\in\F_q$ such that $y_1\neq y_2$ and the $S_1$ and $S_2$ from \cref{lem:exist bad S1 and S2}. Let $S=S_1\cup S_2$, so indeed $\abs{S}=K_{max}$. To count $\abs{\GammaR(S)}$, all we need to do is to subtract the number of 2-wise overlaps of left neighbors $S_1$ and $S_2$ from the total number of possible neighbors of $S_1$ and $S_2$. The latter value is simply $\abs{S}\cdot D_R=K_{max}\cdot D_R$, and the former is given to us by \cref{lem:exist bad S1 and S2} as $\abs{S_1}\cdot\abs{S_2}=\frac{K_{max}^2}{4}$. Therefore, we can compute
    \begin{align*}
        \abs{\GammaR(S)}&=K_{max}\cdot D_R-\frac{K_{max}^2}{4}=K_{max}\cdot D_R\left(1-\frac{K_{max}}{4 D_R}\right),
    \end{align*}
    showing that
    \begin{align*}
        \varepsilon&=\frac{K_{max}}{4D_R}
        =\frac{\delta q^{n-(s+1)}}{4 q^{n-(s+1)}}
        =\frac{\delta}{4}.
    \end{align*}
\end{proof}

The proof of \cref{lem:exist bad S1 and S2} relies on choosing $S_1$ and $S_2$ to contain only the derivatives of polynomials of degree at most $n-(s+1)$.
In the following proofs, we will consider vector spaces of polynomials over $\F_q$.

\begin{definition}
    Define the set $P_d=\{f\in\F_q[x]\mid \deg(f)=d\}$ and let $P_{< d}=\{f\in\F_q[x]\mid \deg(f)< d\}$ be a vector space over $\F_q$ with addition and multiplication coming from $\F_q[x]$. 
\end{definition}

\subsection{Constructing sets with the smallest expansion possible}
We prove \cref{lem:exist bad S1 and S2} by restricting the KT graph to the subgraph containing the two buckets $T_{y_1}$ and $T_{y_2}$ and analyzing the behavior of this subgraph. Consequently, we rephrase \cref{lem:exist bad S1 and S2} as follows.

\begin{lemma}[Technical version of \cref{lem:exist bad S1 and S2}]\label{cor:lem:S1 and S2 construction for psi}
    For any $K\in[q^{n-(s+1)}]$ and distinct $y_1,y_2\in\F_q$, there exist sets $S_1,S_2\subseteq\F_q^{s+1}$ such that $\abs{S_1}+\abs{S_2}=K$ and $S_1\times S_2\in\psi_{y_1,y_2}(P_{<n})$.
\end{lemma}

To create these $S_1$ and $S_2$, we actually first create analogous sets in the image of the Chinese Remainder Theorem map, defined below.

\begin{definition}
    For distinct $y_1,y_2\in\F_q$ let $g_1(x)=(x-y_1)^{s+1}$ and $g_2(x)=(x-y_2)^{s+1}$, let $\pi_1:\F_q[x]\to\F_q[x]/(g_1)$ and $\pi_2:\F_q[x]\to\F_q[x]/(g_2)$ be the associated quotient maps. Then let $\varphi:\F_q[x]/(g_1g_2)\to\F_q[x]/(g_1)\times\F_q[x]/(g_2)$ be defined as $\varphi=\pi_1\otimes\pi_2$. 
\end{definition}
This is exactly the map that the Chinese Remainder Theorem acts on.
\begin{claim}
    The Chinese Remainder Theorem says that $\varphi$ is an isomorphism of rings.\footnote{For an introduction to the Chinese Remainder Theorem, see Chapter 7.6 of \cite{dummit_abstract_2003}.}
\end{claim}
\begin{corollary}
    As a consequence, we may also think of $\varphi$ as an isomorphism of $\F_q$-vector spaces $\varphi:P_{<2s+2}\to P_{<s+1}\times P_{<s+1}$.
\end{corollary}
Whereas $\psi_{y_1,y_2}(f)$ tells us about the first $s$ derivatives of $f$ at $y_1$ and $y_2$, $\varphi(f)$ tells us about $f$ quotiented by $g_1$ and $g_2$. In fact, as we will see, $\varphi$ and $\psi_{y_1,y_2}$ provide us with the same information about a particular polynomial. With this in mind, we prove a lemma similar to \cref{cor:lem:S1 and S2 construction for psi} but for $\varphi$.

\begin{lemma}\label{lem:S1 and S2 construction for phi}
    For any $K\in[q^{n-(s+1)}]$ and distinct $y_1,y_2\in\F_q$, there exist sets $S_1\subseteq\F_q[x]/(g_1)$ and $S_2\subseteq\F_q[x]/(g_2)$ such that $\abs{S_1}+\abs{S_2}=K$ and $S_1\times S_2\in\varphi(P_{<n})$.
\end{lemma}

Using this relation between $\varphi$ and $\psi_{y_1,y_2}$, we prove \cref{cor:lem:S1 and S2 construction for psi}.

\begin{proof}[Proof of \cref{cor:lem:S1 and S2 construction for psi}]
    Use \cref{lem:S1 and S2 construction for phi} to create $R_1\subseteq\F_q[x]/(g_1)$ and $R_2\subseteq\F_q[x]/(g_2)$ such that $\abs{R_1}+\abs{R_2}=K$ and $R_1\times R_2\in\varphi(P_{<n})$. Let $S_1=\psi_{y_1}\circ\pi_1^{-1}(R_1)$ and $S_2=\psi_{y_2}\circ\pi_2^{-1}(R_2)$ where $\pi_1^{-1}$ and $\pi_2^{-1}$ are the natural inclusions of $\F_q[x]/(g_1)$ and $\F_q[x]/(g_2)$ into $\F_q[x]$, respectively. We claim that $S_1\times S_2\in\psi_{y_1,y_2}(P_{<n})$.

    Since the polynomials in the images of $\pi_1^{-1}$ and $\pi_2^{-1}$ are of degree strictly less than $s+1$, we have that $\psi_{y_1}\circ\pi_1^{-1}$ and $\psi_{y_2}\circ\pi_2^{-1}$ are injective, so $\abs{S_1}+\abs{S_2}=\abs{R_1}+\abs{R_2}=K$. Moreover, for a pair $(s_1,s_2)\in S_1\times S_2$, we can take the unique $(r_1,r_2)\in R_1\times R_2$ such that $s_1=\psi_{y_1}\circ\pi_1^{-1}(r_1)$ and $s_2=\psi_{y_2}\circ\pi_2^{-1}(r_2)$. Thus, $(s_1,s_2)=\psi_{y_1,y_2}(\pi^{-1}_1(r_1),\pi_2^{-1}(r_2))\in\psi_{y_1,y_2}(P_{<n})$, as claimed.
\end{proof}

Our proof of \cref{lem:S1 and S2 construction for phi} relies on a result about the structure of the image of $\varphi$ which we prove in \cref{subsec:structure of image of phi} but state here.

\begin{definition}
    For any $b\in\F_q^{s+1}$ define the line $\ell_b=\{(f,f+b)\mid f\in\F_q^{s+1}\}$.
\end{definition}
We now show that the image of $\varphi$ is actually composed of many of these lines with shifts given by the injective homomorphism $\sigma:P_{<n-(s+1)}\to P_{<s+1}$ that we introduce later in \cref{def:sigma}.

\begin{lemma}\label{lem:Structure of image of phi}
We show that:
    \begin{enumerate}
        \item $\varphi(P_{<s+1})=\ell_0$
        \item For $2s+2>d\geq s+1$ we have $\varphi(P_d)=\bigcup_{h\in P_{d-(s+1)}}\ell_{\sigma(h)}$.
    \end{enumerate}
    Thus, we can conclude that $\varphi(P_{<n})=\bigcup_{h\in P_{< n-(s+1)}}\ell_{\sigma(h)}$ where $\sigma$ is an injective homomorphism.
\end{lemma}

This structural result on the image of $\varphi$ allows us to prove \cref{lem:S1 and S2 construction for phi}.

\begin{proof}[Proof of \cref{lem:S1 and S2 construction for phi}]

    Let $S_1,S_2\subseteq\sigma(P_{< n-(s+1)})$ such that $\abs{S_1}=\abs{S_2}=\frac{K_{max}}{2}$. Recall that $\sigma$ is injective so $\abs{\sigma(P_{< n-(s+1)})}=\abs{P_{< n-(s+1)}}=q^{n-(s+1)}$ and we have enough elements to choose from.
    We claim that any such choice of $S_1$ and $S_2$ satisfies the lemma statement.
    
    To prove this claim, we have to show that any pair $(s_1,s_2)\in S_1\times S_2$ lies in $\varphi(P_{<n})$. By \cref{lem:Structure of image of phi}, this is equivalent to saying that $(s_1,s_2)$ is of the form $(f,f+\sigma(h))$ for some $f\in P_{<s+1}$ and $h\in P_{< n-(s+1)}$. From our construction, we have that $s_1=\sigma(h_1)$ and $s_2=\sigma(h_2)$ for $h_1,h_2\in P_{< n-(s+1)}$. 
    So we are considering the point $(\sigma(h_1),\sigma(h_2))$. Using the fact that $\sigma$ is a homomorphism from \cref{claim: sigma is hom}, we can rewrite this point in our desired form:
    \begin{align*}
        (\sigma(h_1),\sigma(h_2))&=(\sigma(h_1),\sigma(h_1)+\sigma(h_2)-\sigma(h_1))\\
        &=(\sigma(h_1),\sigma(h_1)+\sigma(h_2-h_1)).
    \end{align*}
    Thus, $(s_1,s_2)\in\ell_{\sigma(h_2-h_1)}\subseteq\varphi(P_{<n})$, as claimed.
\end{proof}

\subsection{The structure of the image of \texorpdfstring{$\varphi$}{phi}}\label{subsec:structure of image of phi}

We end \cref{sec:tightness} by proving \cref{lem:Structure of image of phi}, that $\varphi(P_{<n})=\bigcup_{h\in P_{< n-(s+1)}}\ell_{\sigma(h)}$. To do so we first define a new homomorphism $\rho$ which we will use to define $\sigma$ later on.

\begin{definition}
    Let $\rho:P_{< n-(s+1)}\to P_{<n-(s+1)}$ be defined as follows. Given $h_1\in P_{<n-(s+1)}$, let $f\in P_{<n}$ be such that $f=h_1g_1+r_1$ for some $r_1\in P_{<s+1}$. Then let $h_2\in P_{\deg(f)-(s+1)}$ and $r_2\in P_{<s+1}$ be the unique polynomials such that $f=h_2g_2+r_2$. We define $\rho(h_1):=h_2$.
\end{definition}
Given this definition, the natural first question is whether $\rho$ is even well-defined since there are $q^{s+1}$ many choices of $f$ that could be used. \cref{lem:h_1 defines h_2} shows that $\rho$ is well-defined and that it is, in fact, just an invertible linear operator, meaning that $\rho$ is an automorphism of $P_{< n-(s+1)}$. 
\begin{lemma}\label{lem:h_1 defines h_2}
    Let $f\in P_{<n}$ and take $h_1,h_2,r_1,r_2\in\F_q[x]$ to be the unique polynomials such that
    \begin{align*}
        f&=h_1g_1+r_1\\
        f&=h_2g_2+r_2
    \end{align*}
    where $\deg(h_1)=\deg(h_2)=\deg(f)-(s+1)$ and $\deg(r_1),\deg(r_2)<s+1$. 

   Then $h_2$ can be determined uniquely from $h_1$, $g_1$, and $g_2$. Similarly, $h_1$ can be determined uniquely from $h_2$, $g_1$, and $g_2$.
\end{lemma}
\begin{proof}[Proof of \cref{lem:h_1 defines h_2}]
    Rearranging the equations in the lemma statement gives us that
    \begin{align*}
        h_1g_1-h_2g_2&=r_2-r_1.
    \end{align*}
    Since the degrees of $r_1$ and $r_2$ are at most $s$, we have that $\deg(r_2-r_1)<s+1$ as well, meaning that $\deg(h_1g_1-h_2g_2)<s+1$. Consequently, the coefficients of $x^k$ for $k\in\{s+1,\dots,n-1\}$ of $h_1g_1$ and $h_2g_2$ must match, which uniquely determines $h_2$ from $h_1$. 

    Formally, if for $t\in\{1,2\}$ we write $h_t(x)=\sum_{k=0}^{n-1-(s+1)}\eta_k^{(t)} x^k$ for $\eta_k^{(t)}\in\F_q$ and expand out $g_t(x)=\sum_{k=0}^{s+1}\gamma_k^{(t)} x^k$  where $\gamma_k^{(t)}\in\F_q$ and $\gamma_{s+1}^{(t)}=1$, then for each $k\in\{s+1,\dots,n-1\}$ we have the linear equation
    \begin{align*}
        \sum_{j=k-(s+1)}^{n-1-(s+1)}\gamma_{k-j}^{(1)}\eta_{j+s+1}^{(1)}&=\sum_{j=k-(s+1)}^{n-1-(s+1)}\gamma_{k-j}^{(2)}\eta_{j+s+1}^{2)}.
    \end{align*}
    This yields $n-s-1$ linear equations, which we can write as $M^{(1)}\eta^{(1)}=M^{(2)}\eta^{(2)}$ where for $t\in\{1,2\}$ we let $\eta^{(t)}=(\eta^{(t)}_0,\dots,\eta^{(t)}_{n-s-2})^\intercal$ and $M^{(t)}\in\F_q^{(n-s-1)\times(n-s-1)}$ is the upper triangular matrix defined as 
    \begin{align*}
        M^{(t)}_{i,j}&=
        \begin{cases}
            \gamma_{s+1+i-j}^{(t)} & i\leq j\\
            0 & i>j
        \end{cases}
    \end{align*}
    Since $g_t$ is monic, meaning that $\gamma_{s+1}^{(t)}=1$, we clearly see that $\det(M^{(t)})=1$, so it is invertible and $\eta^{(2)}=\left(M^{(2)}\right)^{-1}M^{(1)}\eta^{(1)}$. Similarly, $\eta^{(1)}=\left(M^{(1)}\right)^{-1}M^{(2)}\eta^{(2)}$.
\end{proof}

\begin{corollary}\label{lem:rho is module iso}
    $\rho$ is an automorphism of $P_{<n-(s+1)}$.
\end{corollary}
\begin{proof}
    \cref{lem:h_1 defines h_2} shows that we can equivalently define $\rho$ as $\rho=\left(M^{(2)}\right)^{-1}M^{(1)}$, which is an invertible linear transformation.
\end{proof}

Now that we have defined $\rho$ and shown that it is an isomorphism, we use it to build $\sigma$.

\begin{definition}\label{def:sigma}
    Let $\sigma:P_{< n-(s+1)}\to P_{<s+1}$ be defined as $\sigma(h)=hg_1-\rho(h)g_2$. Note that while we consider multiplication here to occur in $\F_q[x]$, \cref{lem:h_1 defines h_2} tells us that this difference yields a polynomial in $P_{<s+1}$.
\end{definition}
Unsurprisingly, $\sigma$ is a homomorphism of vector spaces since $\rho$ is one.
\begin{claim}\label{claim: sigma is hom}
    Since $\rho$ is an isomorphism, $\sigma$ is a homomorphism of vector spaces over $\F_q$.
\end{claim}
\begin{proof}
The fact that $\sigma$ is a homomorphism of vector spaces over $\F_q$ follows from $\rho$ being a homomorphism and addition and multiplication commuting in $\F_q[x]$.
\end{proof}

Finally, we are ready to prove \cref{lem:Structure of image of phi}, that $\varphi(P_{<n})=\bigcup_{h\in P_{< n-(s+1)}}\ell_{\sigma(h)}$.

\begin{proof}[Proof of \cref{lem:Structure of image of phi}]
    We can directly show the first item by the fact that that the remainder of dividing a polynomial by another of larger degree is the polynomial itself. That is, for any polynomial $f\in P_{<s+1}$, we have that $\pi_1(f)=\pi_2(f)=f$, so $\varphi(f)=(f,f)$. Consequently, going over all $f$ of degree less than $s+1$ gives us $\varphi(P_{<s+1})=\{(f,f):f\in P_{<s+1}\}=\ell_0$.

    Now, when considering $P_d$ for $2s+2>d\geq s+1$, quotienting by a degree $s+1$ polynomial does have a non-trivial effect. That is, for $f\in P_d$, there must exist unique $h_1,h_2,r_1,r_2\in\F_q[x]$ such that 
    \begin{align*}
        f&=h_1g_1+r_1\\
        f&=h_2g_2+r_2
    \end{align*}
    where $\deg(h_1)=\deg(h_2)=\deg(f)-(s+1)$ and $\deg(r_1),\deg(r_2)<s+1$. Recall that $g_1(x)=(x-y_1)^{s+1}$ and $g_2(x)=(x-y_2)^{s+1}$. Moreover, $\pi_1(f)=r_1$ and $\pi_2(f)=r_2$. By definition, we have that $h_2=\rho(h_1)$, meaning that $r_2=r_1+h_1g_1-\rho(h_1)g_2=r_1+\sigma(h_1)$. Therefore, we have $\varphi(f)=(r_1,r_1+\sigma(h_1))$. 

    With this in mind, we recall that because $h_1$ and $r_1$ uniquely identify $f$, we can iterate over all elements of $P_d$ by iterating over all $h_1\in P_{d-(s+1)}$ and $r_1\in P_{<s+1}$. That is, $P_d=\{h_1g_1+r_1:h_1\in P_{d-(s+1)},r_1\in P_{<s+1}\}$. From this we can conclude that
    \begin{align*}
        \varphi(P_d)&=\{\varphi(f):f\in P_d\}\\
        &=\{\varphi(h_1g_1+r_1):h_1\in P_{d-(s+1)},r_1\in P_{<s+1}\}\\
        &=\{(r_1,r_1+\sigma(h_1)):h_1\in P_{d-(s+1)},r_1\in P_{<s+1}\}\\
        &=\bigcup_{h_1\in P_{d-(s+1)}}\ell_{\sigma(h_1)},
    \end{align*}
    as claimed.

    Lastly, to see that $\sigma$ is injective we recall that $\varphi$ is an isomorphism and count cardinalities. For the left hand side of $\varphi(P_{<n})=\bigcup_{h\in P_{< n-(s+1)}}\ell_{\sigma(h)}$, we have that $\abs{\varphi(P_{<n})}=\abs{P_{<n}}=q^n$. For the right hand side, we recall that $\abs{\ell_b}=q^{s+1}$ for any $b$, so $\abs{\bigcup_{h\in P_{< n-(s+1)}}\ell_{\sigma(h)}}=\abs{\sigma(P_{< n-(s+1)})}\cdot q^{s+1}$. Therefore, $\abs{\sigma(P_{< n-(s+1)})}=\frac{q^n}{q^{s+1}}=q^{n-(s+1)}$, showing that $\sigma$ is injective since $\abs{P_{<n-(s+1)}}=q^{n-(s+1)}$.
\end{proof}

\section{Non-Bipartite Lossless Expander}\label{sec:non-bipartite lossless expander}
Here, we show how to transform a two-sided lossless expander into an undirected graph (that is not necessarily bipartite) while retaining lossless expansion. We then apply this transformation to the KT graph to obtain our main theorem.
\begin{theorem}[Formal version of \cref{thm: informal non-bipartite pos and impos}]\label{thm:normal lossless expander}
For infinitely many $N$ and all $0 < \delta < 0.99$, there exists an explicit regular $(K, \eps = 0.01)$ lossless expander $G = (V, E)$ where $|V| = N$, the degree is $D$ where $N^{1 - 1.01\delta}\le D \le N^{1 - 1.01\delta + o(1)}$ and
$K = \min\left(N^{\delta}, N^{1 - 1.01\delta - o(1)}\right)$.
Moreover, $G$ with one vertex removed, is endowed with a free group action from $\F_q$, where $q=\poly(\log N)$.
\end{theorem}

\subsection{Expansion from the bipartite half}
Given a two-sided lossless expander, we show how to obtain a (not necessarily bipartite) graph that is also a losslesss expander while inheriting the expansion of this graph. We use the bipartite half transformation defined as follows.

\begin{definition}[Bipartite half]
    Let $G=(L\sqcup R,E)$ be a $(D_L,D_R)$-regular bipartite graph. Then the \emph{bipartite half} $G^2[L]=(L,E^2[L])$ is defined as $E^2[L]=\{(v,w)\in L\times L\mid w\in\GammaR(\GammaL(v))\}$.
\end{definition}

Next, we show how this transformation retains lossless expansion. For the sake of clarity, we will use $\GammaL$ and $\GammaR$ for the left-to-right and right-to-left neighborhood functions of $G$ and $\Gamma$ as the neighborhood function of $G^2[L]$.

\begin{lemma}\label{lem:Two-sided lossless expander to lossless expander}
    Let $G=(L\sqcup R,E)$ be a $(D_L,D_R)$-regular $(K_L,A_L,K_R,A_R)$-two-sided lossless expander with $D_L\leq K_R$. Then $G^2[L]$ is a max-degree $(K,A)$-expander where each node has a degree in $[D_LA_R,D_LD_R]$ and with $K=\min(K_L,K_R/D_L)$ and $A=A_LA_R$.
\end{lemma}
\begin{remark}
    While $G^2[L]$ may not be exactly regular, since $A_L=(1-\varepsilon_L)D_L$ and $A_R=(1-\varepsilon_R)D_R$, we see that $A=A_LA_R=(1-\varepsilon_L)(1-\varepsilon_R)D_LD_R$, meaning that our expansion is with respect to the highest possible degree $D_LD_R$ of any individual vertex.
\end{remark}
\begin{proof}[Proof of \cref{lem:Two-sided lossless expander to lossless expander}]
    We begin by showing that each node $v\in L$ of $G^2[L]$ has degree in $[D_LA_R,D_LD_R]$. By assumption, we have that $\abs{\GammaL(v)}=D_L\leq K_R$. Thus, by the right-to-left expansion of $G$, we have that $\abs{\GammaR(\GammaL(v))}\geq D_LA_R$. The upper bound is immediate given that the right degree is $D_R$ so $\abs{\GammaR(\GammaL(v))}\leq D_R\abs{\GammaL(v)}=D_RD_L$.

    Next, we prove expansion. Let $S\subseteq L$ be a set of size at most $K$. Then, because $K\leq K_L$, the left-to-right expansion of $G$ gives us that $\abs{\GammaL(S)}\geq A_L\abs{S}$. 
    To expand a second time, we recall that $K\leq K_R/D_L$, so $\abs{\GammaL(S)}\leq D_L\abs{S}\leq D_LK\leq K_R$, meaning that we can apply the right-to-left expansion of $G$. 
    This yields $\abs{\GammaR(\GammaL(S))}\geq A_R\abs{\GammaL(S)}\geq A_RA_L\abs{S}$, as claimed. 
\end{proof}

In the special case of the KT graph, the bipartite half is regular. To show this, we make the following observation.

\begin{remark}\label{rmk:bipartite half representation of KT}
    The bipartite half of the KT graph $G$ from \cref{def: KTGraph full} has a succinct representation as $G^2[L]=(L,E^2[L])$ where $E^2[L]=\{(f,g)\mid \exists y\in\F_q,\ \psi_y(f)=\psi_y(g)\}$.
\end{remark}

This allows us to prove the following regularity lemma.
\begin{lemma}\label{lem:KT bipartite half is regular}
    Let $G^2[L]$ be the bipartite half of the KT graph. Then $G^2[L]$ is regular.
\end{lemma}
\begin{proof}
    Let $T_a^n[f](x)=\sum_{i=0}^n\frac{f^{(n)}(a)}{i!}(x-a)^i$ be the $n$-th Taylor polynomial of $f$ at $a$. Then we claim that $\psi_y(f)=\psi_y(g)$ for any $y\in\F_q$ if and only if $T_y^s[f](x)=T_y^s[g](x)$ as polynomials. For the forward direction, we note that $\psi_y(f)=\psi_y(g)$ exactly gives us that $f^{(i)}(y)=g^{(i)}(y)$ for $i\in\{0,\dots,s\}$, immediately implying that $T_y^s[f](x)=T_y^s[g](x)$. Conversely, if $T_y^s[f](x)=T_y^s[g](x)$ as polynomials, then their coefficients must be equivalent. Thus, $f^{(i)}(y)=g^{(i)}(y)$ for $i\in\{0,\dots,s\}$, meaning that $\psi_y(f)=\psi_y(g)$.

    With this claim in hand, we can use \cref{rmk:bipartite half representation of KT} to see that $(f,g)$ is an edge in $G^2[L]$ if and only if there exists some $y\in\F_q$ such that $T_y^s[f](x)=T_y^s[g](x)$. In other words, the neighbors of $f$ must have the same $s$-th Taylor polynomial at some $y\in\F_q$. More formally, the neighbor set of $f$ is $\Gamma(f)=\left\{T_y^s[f](x)+\sum_{i=s+1}^{n-1} a_i(x-y)^i\mid a_{s+1},\dots,a_n,y\in\F_q\right\}$. Thus, the number of neighbors is $\abs{\Gamma(f)}=\abs{\left\{\sum_{i=s+1}^{n-1} a_i(x-y)^i\mid a_{s+1},\dots,a_n,y\in\F_q\right\}}$, which does not depend on $f$. Therefore, the degree of any vertex is the same and $G^2[L]$ is regular.
\end{proof}

\subsection{Free group action on the bipartite half}
Now that we have shown that the bipartite half generally preserves lossless expansion, we will consider it instantiated with the KT graph  and show that multiplication by elements of $\F_q$ constitutes a free group action on this resulting graph (with one node removed).

Our action of $\F_q$ on the bipartite half of the KT graph is directly by multiplication in $\F_q$.
\begin{definition}\label{def:Fq action on KT bipartite half}
    Let $G=(L\sqcup R,E)$ be the KT graph and $G^2[L]$ be its bipartite half. We define the action of $\F_q$ on $G^2[L]$ as follows: for any $\alpha,y\in\F_q$ we have $(\alpha \cdot f)(y)=\alpha \cdot f(y)$ where the latter multiplication is in $\F_q$.
\end{definition}

We now show that $G^2[L]$ without the zero polynomial is $\F_q$-invariant and that this is a free group action.
\begin{lemma}\label{lem:Fq action on KT bipartite half}
    Let $G=(L\sqcup R,E)$ be the KT graph and $H=G^2[L]\setminus\{0\}$ be its bipartite half without the zero polynomial. Consider the action of $\F_q$ on $G^2[L]$ as defined in \cref{def:Fq action on KT bipartite half}. Then $H$ is $\F_q$-invariant and this action is free.
\end{lemma}
\begin{proof}
    To show that $H$ is $\F_q$-invariant, we must prove that for any $(f,g)\in E^2[L]$ and $\alpha\in \F_q$ we have $(\alpha\cdot f,\alpha\cdot g)\in E^2[L]$. From \cref{rmk:bipartite half representation of KT} we know that $E^2[L]=\{(f,g)\mid \exists y\in\F_q,\ \psi_y(f)=\psi_y(g)\}$. Thus, we must equivalently show that if $\psi_y(f)=\psi_y(g)$ for some $y\in\F_q$, then $\psi_y(\alpha\cdot f)=\psi_y(\alpha\cdot g)$. This is immediate by \cref{cor: seedy is fq linear} because $\psi_y$ being $\F_q$-linear allows us to compute
    \begin{align*}
        \psi_y(\alpha\cdot f)=\alpha\psi_y(f)=\alpha\psi_y(g)=\psi_y(\alpha\cdot g),
    \end{align*}
    showing that $(\alpha\cdot f,\alpha\cdot g)\in E^2[L]$.

    Next, we will show that this is a free group action. Consider any $f$ in the vertices of $H$ and $\alpha\in\F_q$. Since $H$ does not contain the zero polynomial, we know that $f$ is not identically zero. Thus, if $\alpha \cdot f=f$, it must be that $\alpha=1$, showing that the action is indeed free.
\end{proof}

\subsection{Plugging in parameters}
Finally, we plug in our two-sided lossless expander result from the KT graph to get \cref{thm:normal lossless expander}.

\begin{proof}[Proof of \cref{thm:normal lossless expander}]
We invoke \cref{thm: instantiate two sided lossless expanders} with $\alpha = 0.01, \eps_L = 0.001, \eps_R = 0.001$ to obtain a $(D_L, D_R)$-biregular two-sided $(K_L = N^{\delta}, \eps_L = 0.001, K_R = 0.002\cdot (1 / D_L)\cdot \min(M, N / M), \eps_R = 0.001)$ lossless expander $\GammaL: [N]\times [D_L]\to [M]$ where $D_L \le O(\log^{204}(N))$ and $N^{1.01\delta - o(1)}\le M\le D_L\cdot N^{1.01\delta}$.
We then apply \cref{lem:Two-sided lossless expander to lossless expander} to conclude the claim about expansion (since $\eps \ge (1 - \eps_L)(1 - \eps_R)$) and the degree bound. We then apply \cref{lem:KT bipartite half is regular} to show the bipartite half is regular. The claim about the free $\F_q$ action comes from \cref{lem:Fq action on KT bipartite half} by removing the vertex corresponding to the zero polynomial from the bipartite half of the KT graph. 
Lastly, we compute that:
\begin{align*}
K 
& = \min\left(K_L, K_R / D_L\right)\\
& = \min\left(K_L, (1 / 500)\cdot (1 / D_L^2)\cdot M,  (1 / 500)\cdot (1 / D_L^2)\cdot (N / M)\right)\\
& = \min\left(N^{\delta}, (1 / 500) \cdot N^{1.01\delta - o(1)},  (1 / 500)\cdot N^{1 - 1.01\delta - o(1)}\right)\\
& = \min\left(N^{\delta}, (1 / 500)\cdot N^{1 - 1.01\delta - o(1)}\right)\\
& = \min\left(N^{\delta}, N^{1 - 1.01\delta - o(1)}\right)\\
\end{align*}
and the bound on maximum size of sets that expand follows.
Explicitness of this graph follows from \cref{claim: explicitness of non-bipartite graph}.
\end{proof}

\section*{Acknowledgements}
We would like to thank Itay Kalev for helpful feedback on our paper, sharing the idea behind \cref{lem:KT bipartite half is regular}, and generously allowing us to include it in the paper. We are also grateful to insightful suggestions from anonymous reviewers that led to simplifying some of our proofs.

\printbibliography

\appendix

\section{Explicitness}\label{sec:expliciteness}

\begin{claim}\label{claim: explicitness of graph}
The KT graph $G$ as defined in \cref{def: KTGraph full} is explicit, i.e., the left neighborhood function $\GammaL: \F_q^n\times \F_q\to \F_q^{s+2}$ and right neighborhood function $\GammaR: \F_q^{s+2}\times \F_q^{n-(s+1)} \to \F_q^n$ can be computed in $\poly(n, \log(q))$ time. 
\end{claim}

\begin{proof}
To compute $\GammaL(f, y)$, we treat $f$ as an element of $\F_q[X]$ of degree at most $n-1$, and map it to $(y, f^{(0)}(y), \dots, f^{(s)}(y))$. We can compute derivatives of $f$ and evaluate it at $y$ in time $\poly(n, \log(q))$ and hence explicitly compute $\GammaL(f, y)$.

To compute $\GammaR(z, t)$, we proceed as follows. Let $z = (y, w)$ where $y\in \F_q$ and $w\in \F_q^{s+1}$.
Then, we need to find $f$ such that $\GammaL(f, y) = z$.
Define $\seed_y: \F_q^n\to \F_q^{s+1}$ as $\seed_y(f) = (f^{(0)}(y), \dots, f^{(s)}(y))$.
By \cref{thm: G right regular and right degree}, $\seed_y$ is a full rank $\F_q$-linear map.
As $n > s+1$, kernel of $\seed_y$ has dimension $n-(s+1) > 0$.
By considering the matrix associated with $\seed_y$ and using standard linear algebra algorithms, we construct an injective linear map $K_y: \F_q^{n - (s+1)} \to \F_q^n$ in $\poly(n, \log(q))$ time such that the image of $K_y$ is exactly the kernel of $\seed_y$.
Furthermore, using Gaussian elimination on the matrix associated with $\seed_y$, we can, in $\poly(n, \log(q))$ time, find some $g\in \F_q^n$ such that $\seed_y(g) = w$.
Finally, we let $\GammaR(z, t) = K_y(t) + g$. By linearity of $\seed_y$, we have that $\seed_y(K_y(t) + g) = \seed_y(g) = w$. As $K_y$ is injective, for a fixed $z$, our computed function maps different $t$ to different outputs as desired.
\end{proof}

Note that given any $n, k, \eps, \alpha$, \cref{thm: instantiate two sided lossless expanders} sets $q = \poly(n, k, 1/\eps)$, so we can deterministically find such a prime $q$ satisfying the requirements in $\poly(n, 1/\eps)$ time as well.

Explicitness of the KT graph also implies explicitness of our non-bipartite lossless expander obtained by taking the bipartite half of the KT graph (and removing the zero vertex):

\begin{claim}\label{claim: explicitness of non-bipartite graph}
The non-bipartite graph $H = G^2[L\setminus \{0\}]$ as constructed in \cref{thm:normal lossless expander} is explicit. I.e., the neighborhood function $\Gamma: (\F_q^n\setminus\{0\})\times (\F_q \times \F_q^{n - (s+1)})\to \F_q^n$ can be computed in $\poly(n, q)$ time. 
\end{claim}

\begin{proof}
Note that since $G$ is not necessarily regular, $\Gamma$ may sometimes output $\bot$.
The guarantee that we will have is that for all $f, g\in \F_q^n$ that are neighbors in $H$, there will exist unique $y, t\in \F_q\times \F_q^{n - (s+1)}$ such that $\Gamma(f, y, t) = g$.

Let $\GammaL, \GammaR$ be the explicit left and right neighborhood functions of the KT graph as defined in \cref{claim: explicitness of graph}.
To compute $\Gamma(f, y, t)$, we first compute $g = \GammaR(\GammaL(f, y), t)$. If $g = 0$, then we output $\bot$. Then, for all $y' < y$, we check whether $\GammaL(f, y') = \GammaL(g, y')$. If they are equal for any such $y'$, we output $\bot$. Otherwise, we output $g$.
This last check is done so that we only output $g$ once as a neighbor of $f$ and otherwise output $\bot$.
Explicitness of $\Gamma$ follows because of explicitness of $\GammaL, \GammaR$ and because the last check has to be done $O(q)$ times.
\end{proof}

\section[The GUV Graph is Not Right Regular]{The \cite{guruswami_unbalanced_2009} Graph is Not Right Regular}\label{sec:GUV not right regular}

One may naturally try to show that the predecessor to the KT graph, the \cite{guruswami_unbalanced_2009} graph, is also a two-sided lossless expander. However, it turns out that the \cite{guruswami_unbalanced_2009} graph is not even right regular. 
To see why, we give the definition of the \cite{guruswami_unbalanced_2009} graph which is similar to the KT graph.
\begin{definition}[The GUV graph, \cite{guruswami_unbalanced_2009}]\label{def: GUVGraph}
    Let $q, n, m, h\in \N$ be such that $q$ is a prime power greater than $h$, characteristic of the finite field $\F_q \ge n$ and $m<n$.  We define $G = (L \sqcup R, E)$ where $L = \F_q^n\cong\F_q[x]/(z(x))$ for some irreducible polynomial $z(x)\in\F_q(x)$ of degree $n$ and $R = \F_q^{m+1}$. The left degree is $q$ and for any $f\in \F_q[x]/(z(x))$ and $y\in \F_q$, the $y$'th neighbor of $f$ is defined as $\GammaL(f,y)=\left(y,f(y),(f^h\text{ mod }z(x))(y),(f^{h^2}\text{ mod }z(x))(y),\dots,(f^{h^{m-1}}\text{ mod }z(x))(y)\right)$.
\end{definition}
Our proof of right-regularity - \cref{thm: G right regular and right degree} - relied on the fact that the map $\psi_y(f)\mapsto(f^{(0)}(y),\dots,f^{(s)}(y))$ is full rank over $\F_q$.
The analogous GUV map $\varphi_y(f)=(f(y),f^h(y),\dots,f^{h^{m-1}}(y))$ does not have this property because of two issues. First, it is not necessarily linear over $\F_q$, although it can be made linear over $\F_2$ when $q$ is a power of 2. Second, even over $\F_2$, it does not necessarily have full rank, meaning we cannot guarantee right regularity.

    Simulations bear this out. The GUV graph with $q=2^4$, $n=4$, $m=2$, and $h=2$ has 3072 right vertices with degree 256, 64 with degree 4096, and 960 isolated vertices. For more examples of parameter settings where the GUV graph is not right-regular, we invite the reader to run simulations with our code at \url{https://github.com/mjguru/GUV-Expander}.

\end{document}